\documentclass[11pt]{article}
\usepackage[margin=0.85in]{geometry}
\usepackage{amsmath, amsthm, amssymb} 
\usepackage{mathtools, thmtools} 

\usepackage{graphicx, subfig}
\usepackage[hidelinks]{hyperref}
\usepackage{cleveref}
\usepackage[mathscr]{euscript} 
\usepackage{xcolor} 
\usepackage{natbib} 
\usepackage{setspace} 
\usepackage{prodint}  
\usepackage[shortlabels]{enumitem}
\usepackage{soul}
\usepackage[titletoc]{appendix}
\usepackage{algorithm, algpseudocode}
\usepackage{thm-restate}
\usepackage{booktabs}
\usepackage{makecell}

\declaretheorem[name=Lemma]{lemma}

\declaretheorem[name=Definition]{define}
\declaretheorem[name=Example]{example}

\renewcommand\thmcontinues[1]{continued}

\DeclareMathOperator*{\argmax}{argmax}
\newcommand*{\tran}{^{\mkern-1.5mu\mathsf{T}}}

\newcommand{\converge}{\xrightarrow{}}

\makeatletter
\def\namedlabel#1#2{\begingroup
    #2%
    \def\@currentlabel{#2}%
    \phantomsection\label{#1}\endgroup
}
\makeatother

\newcommand{\boundeddet}[1]{O(#1)}

\newcommand{\fasterthandet}[1]{o(#1)}

\newcommand{\prob}{\mathrm{P}}

\newcommand{\s}[1]{\mathcal{#1}}
\renewcommand{\d}[1]{\mathbb{#1}} 

\newcommand{\sd}{\,\mathrm{d}}
\newcommand{\E}{\mathrm{E}}

\newcommand{\norm}[1]{\left\Vert #1 \right\Vert}
\newcommand{\abs}[1]{\left\vert#1\right\vert}

\allowdisplaybreaks 

\title{Nonparametric Assessment of Variable Selection \\and Ranking Algorithms}
\author{Zhou Tang \\ zhoutang@umass.edu \and Ted Westling \\ twestling@umass.edu}
\date{Department of Mathematics and Statistics \\
University of Massachusetts Amherst \\
\vspace{1em}
\today}

\begin{document}
\maketitle
\begin{abstract}
Selecting from or ranking a set of candidates variables in terms of their capacity for predicting an outcome of interest is an important task in many scientific fields. A variety of methods for variable selection and ranking have been proposed in the literature. In practice, it can be challenging to know which method is most appropriate for a given dataset. In this article, we propose methods of comparing variable selection and ranking algorithms. We first introduce measures of the quality of variable selection and ranking algorithms. We then define estimators of our proposed measures, and establish asymptotic results for our estimators in the regime where the dimension of the covariates is fixed as the sample size grows. We use our results to conduct large-sample inference for our measures, and we propose a computationally efficient partial bootstrap procedure to potentially improve finite-sample inference. We assess the properties of our proposed methods using numerical studies, and we illustrate our methods with an analysis of data for predicting wine quality from its physicochemical properties. 
\end{abstract}

\doublespacing
\section{Introduction}
\subsection{Background and literature review}
The past few decades have witnessed explosive growth in the amount of data available for analysis. In many settings, researchers may wish to identify a subset of the available variables that are highly predictive of an outcome of interest or to rank the variables in terms of their predictive ability. For example, cell function is often regulated by only a small subset of genes, and identifying which genes control a specific cell function is an important research area in biology \citep{leclerc2008survival}. Similarly, the risk of many medical events only depends on a small subset of the collected covariates, and determining which covariates are risk factors is important for improving scientific understanding of an event of interest \citep{fan2002variable}. Finally, in chemical and materials discovery, feature ranking techniques have emerged as an important tool for evaluating new candidate materials \citep{janet2017resolving}.
 
Due to the importance of variable selection and ranking in an array of scientific fields, a variety of variable selection methods have been proposed. Some of the most well-known methods include best subset selection \citep{hocking1967selection}, least absolute shrinkage and selection operator (LASSO) \citep{tibshirani1996regression}, elastic net \citep{zou2005regularization}, random forest \citep{breiman2001random}, and Bayesian methods \citep{mitchell1988bayesian, george1997approaches, park2008bayesian}. For a broad review of variable selection methods, we refer the reader to \cite{heinze2018variable} and the references therein.  

With a plethora of variable selection and ranking methods to choose from, it is not always clear which of these methods is best for a given dataset. Furthermore, applying different methods to the same data often results in different selected sets or variable ranks. It is therefore important to have a way of choosing the most appropriate selection or ranking method. Different methods have theoretical guarantees that rely on different assumptions about the true data-generating mechanism; e.g., assumptions about the true regression function such as sparsity, linearity, or additivity, or assumptions about the marginal distribution of the predictors such as approximate independence or multivariate normality. If the relevant properties of the true distribution are known, it may be possible to narrow the set of viable methods to those that work well under the given conditions. However, since these properties are typically not known in practice,   it is of interest to compare the performance of variable selection or ranking procedures using the data at hand.

Several authors have studied the problem of comparing the performance of variable selection procedures. \cite{heinze2018variable} suggested assessing the stability and sensitivity of candidate algorithms using quantities such as inclusion frequencies and root mean squared difference over bootstrap samples. Other authors have employed cross-validation \citep{stone1974cross, lachenbruch1968estimation, cover1969learning} to compare variable selection procedures. By comparing the cross-validated area under the ROC curve (AUC)  and the number of variables selected from a number of methods, \cite{sanchez2018comparison} found that classic regression-based variable selection methods are more suitable for smaller sample sizes, while tree-based methods perform better with more data.  \citet{murtaugh2009performance} and \cite{refaeilzadeh2007comparison} suggested using cross-validated risk, such as mean squared error, as a way to compare variable selection algorithms.

It is also important to quantify uncertainty when comparing variable selection and ranking procedures. For example, while selecting more variables typically reduces the cross-validated risk, a confidence interval may reveal that there is substantial uncertainty in the risk reduction. In this case, the method that selects a more parsimonious model may be preferred. Alternatively, if two different methods select different subsets of the same size, testing the null hypothesis that the two sets have the same risk or obtaining a confidence interval for the difference or ratio of the risks allows the researcher to rigorously assess whether and to what extent one subset is more predictive than the other. Finally, when comparing variable ranking procedures across a range of subset sizes, uniform confidence bands for the risks or risk differences over the range allow the researcher to assess whether one method dominates another. To the best of our knowledge, no method yet exists for providing valid inference for comparisons of variable ranking procedures.

In this article, we propose a framework for empirically comparing variable selection and ranking procedures. Our framework is based on the nonparametric methods of assessing variable importance proposed in \cite{williamson2022general}. However, \cite{williamson2022general} considered assessing the importance of fixed sets of variables, whereas here, the variable sets are random depending on the data because the selection and ranking procedures use the data in potentially complicated ways. In Section~\ref{sec:methods}, we review the framework of~\cite{williamson2022general}, define our proposed measures of the quality of a variable selection or ranking procedure, and discuss their interpretation. In Section~\ref{sec:theory}, we propose estimators of our proposed measures and provide conditions under which our estimators are asymptotically linear when the dimension of the covariates is fixed as the sample size grows. In Section~\ref{sec:inference}, we use our asymptotic results to derive large-sample confidence regions for our parameters of interest, and we also propose a computationally efficient modified bootstrap procedure to potentially improve finite-sample inference. In Section~\ref{sec:sim}, we present a simulation study assessing the finite-sample properties of our methods, and in Section~\ref{sec:analysis}, we use the proposed methods to compare variable selection and ranking methods for predicting wine quality from its physicochemical properties. In Section~\ref{sec:conclusion}, we provide a brief discussion. The proofs of all theorems are provided in Supplementary Material.

\section{Parameters of interest and their interpretation}\label{sec:methods}

\subsection{Statistical setting}

We suppose that $X = (X_1, \dotsc, X_p) \in \s{X} \subseteq \d{R}^p$ is a $p$-dimensional covariate vector, and $Y \in \s{Y} \subseteq \d{R}$ is a scalar outcome. We suppose that the observed data consist of $n$ independent and identically distributed vectors $\{(X_{1i}, \dotsc, X_{pi}, Y_i): i = 1, \dotsc, n\}$ drawn from an unknown distribution $P_0$ on the sample space $\s{X} \times \s{Y}$. We assume that $P_0$ is known to lie in a model $\s{M}$, which is typically a nonparametric model. With some abuse of notation, we also use $P_0$ as the true marginal distribution of $(X_{1i}, \dotsc, X_{pi})$.  Throughout, we focus on the setting where the number of covariates $p$ is fixed with sample size. The use of subscript $0$ refers to evaluation at or under $P_0$; for example, we write $E_0$ to denote expectation under $P_0$. We define $\d{P}_n$ as the empirical distribution of $\{(X_{1i}, \dotsc, X_{pi}, Y_i): i = 1, \dotsc, n\}$. For any measure $P$ and $P$-integrable function $f$, we set $P f := \int f \sd P$. For any $S \subseteq \{1, \dotsc, p\}$ and $X \in \s{X}$, we denote $X_{-S}$ as the elements of $X$ whose indices \textit{do not} fall in $S$, and we let $\s{X}_{-S}$ be the sample space of $X_{-S}$.

\subsection{Measures of variable importance}\label{sec:vimp}
As discussed in the introduction, \cite{williamson2022general} proposed a unified approach to defining algorithm-agnostic, population-level measures of predictiveness of a subset of covariates. We will use their framework as part of our method of quantifying the quality of a variable selection or ranking algorithm, so we briefly describe the key elements of their approach.

The approach of \cite{williamson2022general} requires a user-defined \emph{predictiveness metric} $V : \s{F} \times \s{M} \to \d{R}$, where $\s{F}$ is a class of functions from $\s{X}$ to $\d{R}$ endowed with a norm $\norm{\cdot}_{\s{F}}$. For example, $\s{F}$ may consist of all $f : \s{X} \to \d{R}$ such that $\int f(x)^2 \sd P_0(x) < \infty$ and $\norm{f}_{\s{F}} := [\int f(x)^2 \sd P_0(x) ]^{1/2}$. Then, $V(f, P)$ is assumed to provide a measure of the predictiveness of a candidate prediction function $f \in \s{F}$ when generating data from $P \in \s{M}$, where higher values are assumed to correspond to better predictiveness. We assume that $V(f,P) \in [0,1]$ for all $(f, P) \in \s{F} \times \s{M}$. The \emph{population maximizer} $f_0 \in \argmax_{f \in \s{F}} V(f, P_0)$, is the best possible prediction function in $\s{F}$ relative to $V$ under sampling from $P_0$, and the \emph{oracle predictiveness} $V(f_0, P_0)$ measures the best possible capacity of the entire covariate set $X$ for predicting $Y$ under sampling from $P_0$. Given  $S \subset \{1, \dotsc, p\}$, \cite{williamson2022general} analogously define the \emph{residual oracle predictiveness} as  $V(f_{0, -S}, P_0)$, where $f_{0, -S} \in \argmax_{f \in \s{F}_{-S}} V(f, P_0)$ and $\s{F}_{-S}$ is the subset of $\s{F}$ consisting of functions $f \in \s{F}$ that \emph{do not} depend on the covariates with indices in $S$. For example, if we have three covariates $X_1, X_2, X_3$ and $S = \{2,3\}$, then $\s{F}_{-S}$ consists of functions in $\s{F}$ that only depend on $X_1$. Thus, $V(f_{0, -S}, P_0)$ quantifies the remaining prediction potential after excluding the covariates with indices in $S$. The \emph{population variable importance}, defined as  $\psi_{0, S}:= V(f_{0}, P_0) - V(f_{0, -S}, P_0)$, measures the amount of oracle predictiveness lost by excluding covariates with indices in $S$. 

In this paper, we make a simplifying assumption about the form of the predictiveness metric $V$. We assume there exist  $\zeta: \s{F} \times \s{M} \mapsto \d{R}$ and $\eta : \s{M} \to \d{R}$ and a function $U : \mathrm{range}(\zeta) \times \mathrm{range}(\eta) \to \d{R}$ such that $P \mapsto \zeta(f, P)$ is linear and 
\begin{equation}\label{eq:Vform}
    V(f, P) = U\left(\zeta(f, P), \eta(P) \right).
\end{equation}
Hence, we assume that $V$ only depends on $f$ and $P$ together through a function that is linear in $P$. Such $V$ are referred to as \emph{standardized $V$-measures of degree one} in \cite{williamson2022general}. Additional conditions on $\zeta$, $\eta$, and $U$ will be provided in Section~\ref{sec:efficiency}. Among the four examples of predictiveness metrics $V$ considered in \cite{williamson2022general}, only one, the area under the ROC curve, does not have this form. Furthermore, we have found that this form greatly simplifies the technical conditions used for the theoretical results provided in Section~\ref{sec:theory}. Additional discussion of the simplifying nature of this assumption is provided following Theorem~\ref{thm:asym linearity}. Hence,  we have decided to sacrifice some generality in making this assumption about the form of $V$ for the sake of clarity and simplicity.

We now review three examples of predictiveness measures $V$ studied in \cite{williamson2022general}.
\begin{example}[$R$-squared]\label{example:rsquared}
    The predictiveness metric $V$ is given by $V(f, P) := 1 - E_P[Y-f(X)]^2 / \sigma^2_P$, and where $\sigma^2_P := E_P [Y - E_P (Y)]^2$ is the marginal variance of $Y$ under $P$. This measure quantifies the proportion of variance in $Y$ explained by $f(X)$. In this case, we have $\zeta(f, P) = E_P[Y-f(X)]^2$, $\eta(P) = E_P [Y - E_P (Y)]^2$ and $U(v, w) = 1-v/w$.
\end{example}
\begin{example}[Deviance]\label{example:deviance}
    Suppose $Y$ is binary, and define $V(f, P) := 1 - E_P[\nu(Y, f(X))] / \nu(\pi_P, \pi_P)$ for $\nu(u,v) := u \log v + (1-u) \log (1-v)$, and where $\pi_P := P(Y=1)$. This measure assesses the goodness of fit of the model fitted by $X$. In this case, we have $\zeta(f, P) = E_P[\nu(Y, f(X))]$, $\eta(P) = \nu(\pi_P, \pi_P)$ and $U(v, w) = 1-v/w$.
\end{example}
\begin{example}[Classification accuracy]\label{example:classification}
    Suppose again that $Y$ is binary, and define $V(f, P) := P(Y = f(X))$. This measure quantifies how often the prediction $f(X)$ coincides with $Y$. In this case, we have $\zeta(f, P) = P(Y = f(X))$, $\eta(P) = 1$ and $U(v, w) = v$.
\end{example}

In Examples~\ref{example:rsquared}--\ref{example:deviance}, the maximizer of $f \mapsto V(f, P)$ over all $f : \s{X} \to \d{R}$ in $\s{F}$ is the conditional mean function $\mu_P : x \mapsto E_P( Y \mid X = x)$, and in Example~\ref{example:classification}, it is the Bayes classifier $x \mapsto I\{ \mu_P(x) > 0.5\}$. Hence, if $\mu_0 \in \s{F}$, then $f_0 = \mu_0$ in the Examples~\ref{example:rsquared}--\ref{example:deviance}, and $f_0 = I\{ \mu_0 > 0.5\}$ in Example~\ref{example:classification}. Similarly, the maximizer of $f \mapsto V(f,P)$ over  $f \in \s{F}_{-S}$ is $\mu_{P, -S} : x \mapsto E_P(Y \mid X_{-S} = x_{-S})$ in Examples~\ref{example:rsquared}--\ref{example:deviance} and $x \mapsto I\{ \mu_{P, -S}(x) > 0.5\}$ in Example~\ref{example:classification}.

A central feature of this variable importance framework is that it is \emph{algorithm-agnostic}, meaning that the definition of the population variable importance is not dependent on the use of a particular algorithm for estimating $f_0$ or $f_{0, -S}$. Valid inference for the variable importance parameter using the observed data requires estimating $f_0$ and $f_{0, -S}$, which does involve choosing an estimation algorithm, and \cite{williamson2022general} provided high-level conditions on the properties of these estimators that yield asymptotically valid inference for the variable importance. However, the choice of an algorithm for purposes of estimation of the variable importance is separate from the definition of the parameter itself. We adopt this same approach when defining our measures of variable selection and ranking procedures.

\subsection{Measures of variable selection algorithms}
\label{sec:variable selection algo}

\cite{williamson2022general} focused on assessing the importance of a fixed set of covariates $S$ defined \emph{a priori} by the researcher. Here, our focus is on assessing and comparing the performance of automatic variable selection algorithms---that is, algorithms that use the observed data to select a subset of the covariates. The variables selected by such an algorithm are a random subset of the covariates, which we denote $S_n \subseteq \{1, \dotsc, p\}$. Here, the subscript $n$ indicates that this random subset depends on the $n$ observed data points in some possibly complicated way. In the spirit of \cite{williamson2022general}, to produce an interpretable measure of the quality of a variable selection procedure, we first define a population parameter of interest for a given variable selection procedure, and we then tackle statistical inference for this population parameter. Like \cite{williamson2022general}, our measures will be algorithm-agnostic in the sense that they will not be tied to the model or algorithm used by a given selection or ranking procedure. 

In order to define a population parameter of interest, the simplest approach is to assume that the random subset $S_n$  converges in probability to a fixed $S_0 \subseteq \{1, \dotsc, p\}$. We will relax this condition later. We can then consider $S_n$ as an estimator of $S_0$. We then define our population-level parameter of interest as the variable importance $\psi_{0, S_0} = V(f_{0}, P_0) - V(f_{0, -S_0}, P_0)$ of the limiting subset $S_0$.  While the population parameter follows readily from the work of \cite{williamson2022general}, we note that asymptotic results regarding statistical inference require extra work due to the randomness of $S_n$.

As discussed in the introduction, one of our goals is to compare the quality of various variable selection algorithms for a given data set. Given two random subsets $S_{n}$ and $S_{n}'$ produced by two different variable selection algorithms, we can consider the population variable importances $\psi_{0, S_{0}}$ and $\psi_{0, S_{0}'}$ as our parameters of interest, where $S_{0}$ and $S_{0}'$ are the limiting subsets to which $S_{n}$ and $S_{n}'$ are assumed to converge. However, this natural approach suffers an important drawback. Variable importance metrics are nested: $S_0 \subset S_0'$ implies that $\psi_{0, S_0} \leq \psi_{0, S_0'}$. Hence, simply comparing  $\psi_{0, S_{0}}, \psi_{0, S_{0}'}$ is not sufficient to compare the quality of the selection algorithms (even if these population quantities were known exactly), because an algorithm that tends to select more variables will tend to produce higher variable importance. As an extreme example, the trivial algorithm that always selects all the variables will always have the maximal possible variable importance.

To resolve this drawback, we propose simply adding a second piece of information: the number of variables selected by the algorithm. Thus, our bivariate parameter of interest is $(|S_0|, \psi_{0,S_0})$. Furthermore, different algorithms can then be graphically compared by plotting the bivariate parameter  in the coordinate plane $[1,p] \times [0,1]$. This plot conveys how predictive the subset selected by each algorithm is against the number of covariates selected. If an algorithm achieves high variable importance with few selected variables, the point corresponding to the true parameter vector of the algorithm would be in the upper left region of the plot. For two limiting subsets $S_0$ and $S_0'$, if the point $(|S_0|, \psi_{0, S_0})$ is to the lower right of the point $(|S_0'|, \psi_{0, S_0'})$ on this plot, meaning that $|S_{0}| \geq |S_{0}'|$ and $\psi_{0, S_{0}} \leq \psi_{0, S_{0}'}$, and at least one of these inequalities is strict, then an algorithm with limiting subset $S_{0}'$ dominates an algorithm with limiting subset $S_{0}$ because its selected subset is smaller, yet has higher variable importance. On the other hand, if $|S_{0}| > |S_{0}'|$ and $\psi_{0, S_{0}} > \psi_{0, S_{0}'}$, then the subset $S_0$ is more important, but also larger, than the subset $S_{0}'$. In this case, whether the additional variables are worth the gain in importance is up to the user. One simple way to combine the two pieces of information to produce a single metric is by dividing the variable importance of the selected set by the size of the selected set; i.e.\ $\psi_{0, S_{0}} / |S_{0}|$. This measure can be interpreted as the average variable importance per selected variable, and graphically, can also be portrayed  in the suggested diagram as the slope of the line connecting $(0,0)$ and $(|S_{0}|, \psi_{0, S_{0}})$. We  call this parameter the \textit{predictiveness per selected variable} (PPSV) for short. An illustration of this diagram will be given in Section~\ref{sec:illust}.

\subsection{Measures of variable ranking algorithms}\label{sec: ranking}

We now extend the population parameters proposed in Section~\ref{sec:variable selection algo} for comparing variable selection algorithms to compare \emph{variable ranking algorithms}; i.e., algorithms that rank the $p$ variables in terms of their potential for predicting the outcome. A variable ranking algorithm is a random ranking of the covariates. Specifically, we define a variable ranking algorithm $R_n$ based on the $n$ data points as a random permutation of  $\{1, \dotsc, p\}$. For each $j \in \{1, \dotsc, p\}$, we denote $[j]:= \{1, \dotsc, j\}$ and $R_{n,[j]}$ as the first $j$ elements in $R_n$, which are the indices of the $j$ \emph{most predictive} covariates according to the algorithm. For example, if $p = 3$, and $R_n = (3,2,1)$, then $R_n$ ranks $X_3$ as the most important for predicting $Y$, $X_2$ as the second most important, and $X_1$ as the least important, and $R_{n,[2]}$ would be $(3,2)$. As with variable selection algorithms, we are most interested in situations where $R_n$ depends on the data in a possibly complicated way. For instance, $R_n$ may be the random variable ranking resulting  from running a penalized regression algorithm such as the least absolute shrinkage and selection operator (LASSO) \citep{tibshirani1996regression}, elastic net \citep{zou2005regularization}, or the smoothly clipped absolute deviation (SCAD) \citep{fan2001variable} algorithms on the data.

In order to define a population parameter of interest, the simplest approach is again to assume that the random ranking $R_n$ converges in probability to a fixed rank $R_0$; i.e., $P_0 (R_{n} = R_{0}) \converge 1$.  This too will be relaxed later. We then consider $R_n$ as an estimator of $R_0$. Intuitively, given a predictiveness metric $V$, the quality of a variable ranking algorithm is higher if the variables that it tends to rank first have higher variable importance. Hence, for any variable rank $R_0$, we define the population \emph{variable ranking operator characteristic} (VROC) as
\begin{equation}
    \label{vim sequence}
    \left(\psi_{0, R_{0, [1]}}, \dotsc, \psi_{0, R_{0, [p]}} \right) = \left(V(f_{0}, P_0) -V(f_{0, -R_{0, [1]}}, P_0), \dotsc, V(f_{0}, P_0) - V(f_{0, -R_{0, [p]}}, P_0) \right),
\end{equation}
so that $\psi_{0, R_{0, [j]}} = V(f_{0}, P_0) - V(f_{0, -R_{0, [j]}}, P_0)$ is the population variable importance of the variables indexed by $R_{0, [j]}$.  We note that $\psi_{0, R_{0, [1]}} \leq \psi_{0, R_{0, [2]}} \leq \cdots \leq \psi_{0, R_{0, [p]}}$ because adding variables to the set used for prediction increases the size of the subset of $\s{F}$ over which the optimization occurs.

We call the curve formed by plotting $\left(\psi_{0, R_{0, [1]}}, \dotsc, \psi_{0, R_{0, [p]}} \right)$ on the vertical axis against $(1,\dotsc, p)$ on the horizontal axis the \emph{VROC curve}. Examples of population VROC curves will be provided in Section~\ref{sec:illust}. As with the ROC curve in the context of prediction of a binary outcome \citep{woodward1953probability}, the closer the VROC curve is to the upper left corner of the rectangle $[1,p] \times [0, 1]$, the better the performance of the ranking algorithm. This is because the ideal ranking algorithm ranks the variables with the largest impact on the predictiveness metric first. Unlike the ROC curve, which necessarily ends at the point $(1, 1)$, the endpoint of the VROC curve is typically not $(p,1)$ because the oracle predictiveness of all variables is usually not 1. However, the endpoint of the VROC curve is the same for all possible variable rankings because the endpoint represents the variable importance of all of $p$ covariates. In addition, unlike the ROC curve, there is not necessarily an optimal variable ranking in terms of the VROC curve in the sense that there may not exist any ranking such that the corresponding VROC curve is no smaller than that of any other ranking at all points. An exception is the case where the true conditional mean function $\mu$ is additive in the covariates and $V$ is the $R$-squared predictiveness metric. In this case, the VROC curve is optimal if and only if it is concave. This will be demonstrated more in the examples below. 

A one-number measure of the overall population performance of a ranking algorithm can be obtained as $\phi(\psi_{0, R_{0, [1]}}, \dotsc, \psi_{0, R_{0, [p]}})$ for any order-preserving summary $\phi: \d{R}^p \to \d{R}$. Here, by order-preserving, we mean that $a_j \leq b_j$ for each $j \in \{1, \dotsc, p\}$ implies that $\phi(a_1, \dotsc, a_p) \leq \phi(b_1, \dotsc, b_p)$. In particular, we will consider 
\[
    \phi(\psi_{0, R_{0, [1]}}, \dotsc, \psi_{0, R_{0, [p]}}) := \sum_{j=2}^{p}\frac{\psi_{0, R_{0, [j]}} + \psi_{0, R_{0, [j-1]}}}{2},
\]
which measures the area under the linear interpolation of the VROC curve. We call this the \emph{the area under the VROC curve} (AUVROC) for short.

\subsection{Illustrative example}\label{sec:illust}

Here, we illustrate the proposed measures of variable selection and ranking algorithms in a toy example. We suppose there are $p = 10$ covariates drawn from an independent uniform distributions on $[-1, 1]$. The outcome $Y$ is generated according to $Y = 0.4X_1 + \sqrt{X_2+1} + 2X_3^2 + \varepsilon$, where $\varepsilon$ follows a standard normal distribution independent of the covariates. We use the $R$-squared predictiveness metric $V$ defined in Section~\ref{sec:vimp}. We consider two variable selection algorithms: the variable subset returned by marginal regression (MR), and the subset selected by multivariate adaptive regression splines (MARS) \citep{friedman1991multivariate}. For MR, we regress the outcome on each covariate separately using univariate linear regression and select the covariates whose resulting absolute standardized coefficients exceed a threshold. In this example, we use $0.1$ and $0.2$ as two thresholds for MR, and the corresponding selection algorithms are denoted MR(0.1) and MR(0.2). MARS uses regression splines to automatically model nonlinearities and interactions between variables.

\begin{figure}[!t]
    \centering
    \includegraphics[width = 0.49\linewidth]{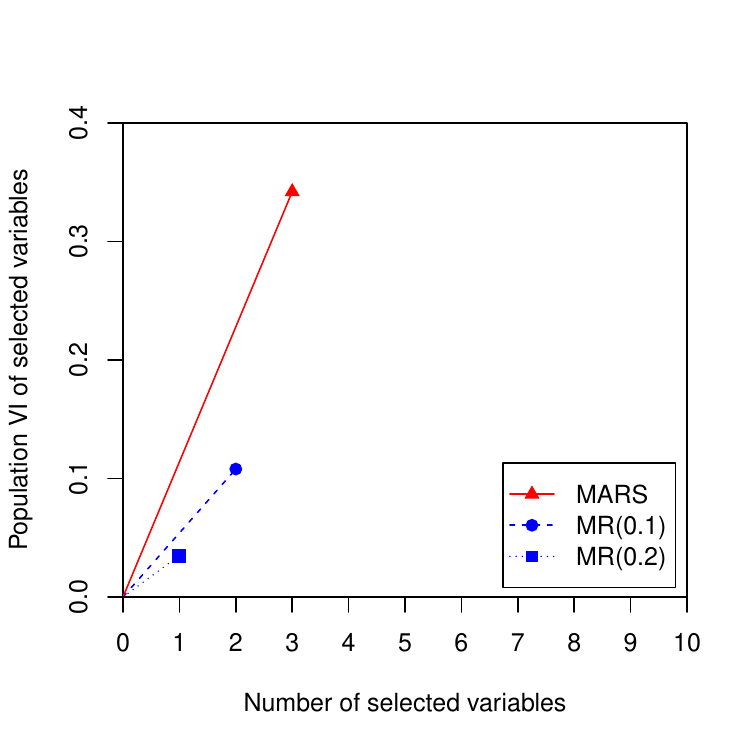}
    \includegraphics[width = 0.49\linewidth]{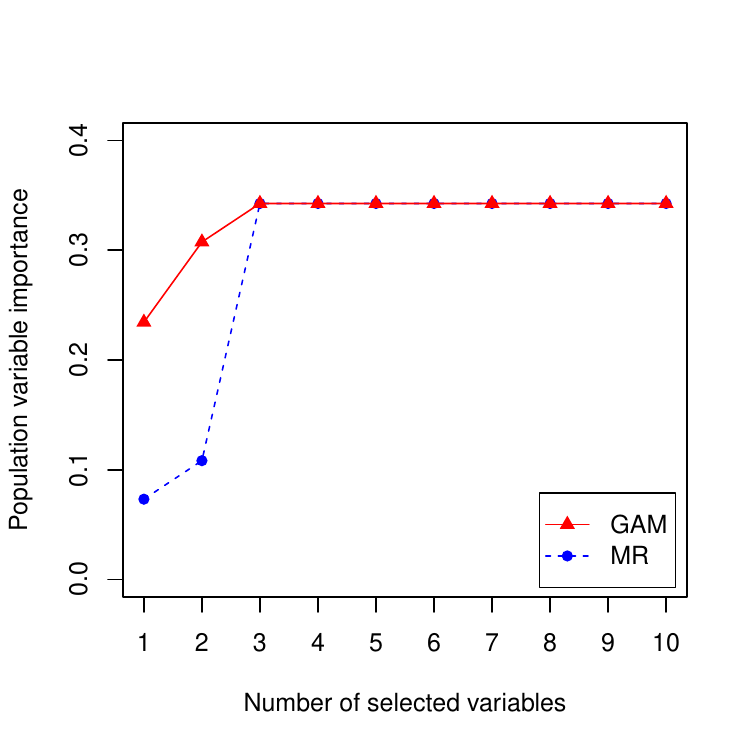}
    \caption{Illustration of the variable selection and ranking measures using the example described in Section~\ref{sec:illust}. Left: Variable importances of selected variables corresponding to the MARS, MR(0.1), and MR(0.2) selection algorithms. Right: VROC curves corresponding to the MR and GAM ranking algorithms.}
    \label{fig:illustration}
\end{figure}

By simulation, we find the limiting selected variables for MR(0.1), MR(0.2), and MARS are $S_{0, \text{MR}(0.1)} = \{1, 2\}$, $S_{0, \text{MR}(0.2)} = \{2\}$ and $S_{0, \text{MARS}} = \{1, 2, 3\}$. The corresponding population bivariate parameters of interest are $(|S_{0, \text{MR}(0.1)}|, \psi_{0, S_{0, \text{MR}(0.1)}}) = (2, 0.11)$, $(|S_{0, \text{MR}(0.2)}|, \psi_{0, S_{0, \text{MR}(0.2)}}) = (1, 0.035)$ and $(|S_{0, \text{MARS}}|, \psi_{0, S_{0, \text{MARS}}}) = (3, 0.34)$ respectively. Figure~\ref{fig:illustration} (left panel) shows the three resulting bivariate parameters of interest in the rectangle $[1,p] \times [0,1]$. Although the limiting subset selected by MARS has higher variable importance, MARS doesn't dominate MR(0.1) or MR(0.2) because these selected fewer variables than MARS. However, the PPSV corresponding to MARS, MR(0.1), and MR(0.2) are $0.11$ and $0.054$ and $0.035$, respectively, so the PPSV of MARS is larger than either of the MR methods.  This can be determined using the left panel of  Figure~\ref{fig:illustration} by noting that the slope of the line corresponding to MARS is larger than that of MR(0.1) and MR(0.2).

To illustrate our proposed measures of variable ranking algorithms, we compare two ranking procedures: the rank obtained by sorting the absolute standardized coefficients from marginal regression (MR), and the rank obtained by sorting the $p$-values from smallest to largest from a generalized additive model (GAM) \citep{hastie1990generalized} of $Y$ on the covariates (using the default settings in the \texttt{mgcv} package in \texttt{R}). By simulation, we find that the limiting rankings of MR and GAM from most important to least important in predicting the outcome differ only on the first three variables. MR returns  $X_2, X_1, X_3$ as the order of the first three variables, and GAM returns $X_3, X_2, X_1$ as the order of the first three important variables. Correspondingly, the population variable importances are $(0.073, 0.11, 0.34, \dotsc, 0.34)$ and $(0.23, 0.31, 0.34, \dotsc, 0.34)$ respectively. Figure~\ref{fig:illustration} (right panel) shows the two resulting VROC curves for MR and GAM with values of AUVROCs equal to 2.7 and 3.0, respectively. In this case, the population ranking of GAM is better than that of MR because the first and first two variables selected by GAM have higher population $R$-squared values than the corresponding variables selected by MR. This is because the true data-generating process follows a GAM, but not a linear model, and MR fails to recognize the importance of the covariates with non-linear effects. We also note that since the true model is a GAM, the optimal VROC curve is concave, as discussed above. An example where the optimal VROC curve is not concave will be provided in Section~\ref{sec:sim}.

\section{Estimation and asymptotic results}\label{sec:theory}

\subsection{Efficiency theory}\label{sec:efficiency}

In this section, we introduce estimators of our parameters of interest and provide theoretical results guaranteeing convergence in distribution of our proposed estimators to mean-zero normal distributions. We use these results to construct asymptotically valid confidence sets for our parameters of interest based on our estimators.

We begin with conditions under which  the parameter mappings $P \mapsto \psi_{P, S_P}$ and $P \mapsto \psi_{P, R_{P,[j]}}$ are pathwise differentiable relative to a nonparametric model for each $j \in \{1, \dotsc, p\}$,  where $R_P$ is the limiting variable rank and $S_P$ is the limiting selected subset under $P \in \s{M}$. Pathwise differentiability of a parameter is important for several reasons. First, it provides a generalized Cr{\'{a}}mer-Rao lower bound for estimating the parameter in a nonparametric or semi-parametric model. Second, under additional regularity conditions, it guarantees that it is possible to construct an \emph{asymptotically linear} estimator of the target parameter that achieves the lower bound \citep{pfanzagl1982contributions, klaassen1987consistent, pfanzagl1990estimation, van1991differentiable}. Asymptotic linearity of an estimator implies $n^{-1/2}$-rate consistency and asymptotic normality, which facilities large-sample statistical inference. Asymptotic linearity further enables joint asymptotic results, which facilities joint large-sample inference. For a review of pathwise differentiability, asymptotic linearity, and semiparametric efficiency theory we refer the reader to \cite{le1990asymptotics, van2000asymptotic} and \cite{kennedy2016semiparametric}.

We recall that we assume $V(f, P) = U( \zeta(f, P), \eta(P))$. We introduce the following conditions for pathwise differentiability, which are specific to a subset $S \subseteq \{1, \dotsc, p\}$. We define $L_2^0(P_0)$ as $\{ h \in L_2(P_0) : P_0 h = 0\}$.
\begin{enumerate}[label=\textbf{(A\arabic*)},leftmargin=2cm]    
    \item[\namedlabel{cond: local continuity}{(A1)}] There exists a function $\dot\zeta : \s{F}  \to L_2(P_0)$  such that $\zeta(f, P) = P [ \dot\zeta(f) ]$ and such that $f \mapsto \dot{\zeta}(f)$ is continuous at $f_{0, -S}$. Also, there exist $C, \delta \in (0, \infty)$ such that for all $f \in \s{F}_{-S}$ with $\norm{f - f_{0, -S}}_{\s{F}} < \delta$,
    \[
        \abs{P_0 \left[ \dot\zeta(f) - \dot\zeta(f_{0,-S}) \right] } \leq C\norm{f - f_{0, -S}}_{\s{F}}^2.
    \]
    \item[\namedlabel{cond: nuisance}{(A2)}] The map $P \mapsto \eta(P)$ is pathwise differentiable at $P_0$ relative to a nonparametric model with nonparametric efficient influence function at $P_0$ given by $\phi_0 \in L_2^0(P_0)$.
    \item[\namedlabel{cond: derivative}{(A3)}] $(\zeta, \eta) \mapsto U(\zeta, \eta)$ is differentiable at $(\zeta, \eta) = (\zeta( f_{0,-S}, P_0), \eta_0)$ with $\dot{U}_{\zeta}(\zeta, \eta) := \frac{\partial U}{\partial\zeta}(\zeta, \eta)$ and $\dot{U}_{\eta}(\zeta, \eta)  := \frac{\partial U}{\partial\eta}(\zeta, \eta)$.
\end{enumerate}

Under these conditions, we have the following result.
\begin{restatable}{thm}{thmpathdiff}\label{thm:pathwise differentiable}
    If conditions~\ref{cond: local continuity}-\ref{cond: derivative} hold for $S=S_0$,  then $P \mapsto V(f_{P, -S_P}, P)$ is pathwise differentiable at $P_0$ relative to a nonparametric model with nonparametric efficient influence function
    \begin{align*}
         \dot{V}_0(f_{0, -S_0}) = \dot{U}_{\zeta}(\zeta(f_{0,-S_0}, P_0), \eta_0) \left[\dot{\zeta}(f_{0, -S_0}) - P_0\dot{\zeta}(f_{0, -S_0})\right] + \dot{U}_{\eta}(\zeta(f_{0,-S_0}, P_0), \eta_0) \phi_0.
    \end{align*}
    If conditions \ref{cond: local continuity}-\ref{cond: derivative} hold  for each $S = R_{0, [j]}$ where $j \in \{1, \dotsc, p\}$, $P \mapsto V(f_{P, -R_{P, [j]}}, P)$ is pathwise differentiable at $P_0$ relative to a nonparametric model with nonparametric efficient influence function 
    \begin{align*}
         \dot{V}_0(f_{0, -R_{0, [j]}}) = \dot{U}_{\zeta}(\zeta(f_{0,-R_{0, [j]}}, P_0), \eta_0) \left[\dot{\zeta}(f_{0, -R_{0, [j]}}) - P_0\dot{\zeta}(f_{0, -R_{0, [j]}})\right] + \dot{U}_{\eta}(\zeta(f_{0,-R_{0, [j]}}, P_0), \eta_0) \phi_0.
    \end{align*}
\end{restatable}
The proof of Theorem~\ref{thm:pathwise differentiable} and  all other theorems are provided in Supplementary Material. Theorem~\ref{thm:pathwise differentiable} provides the means to  construct asymptotically linear estimators of our measures of variable selection and ranking algorithms, as we will see below. Condition~\ref{cond: local continuity} requires linearity of $P \mapsto \zeta(f, P)$. As mentioned earlier, this holds in many examples, and simplifies many technical details. Condition~\ref{cond: local continuity} also requires that the local behavior of $f \mapsto V(f, P_0)$ in a neighborhood of $f = f_{0, -S}$ is controlled by the quadratic norm of $f - f_{0, -S}$. This can be expected to hold because, as discussed in \cite{williamson2022general} and elsewhere, $f_{0,-S}$ is defined an optimizer of $f \mapsto V(f, P_0)$. Condition~\ref{cond: nuisance} requires that the parameter $\eta(P)$ is pathwise differentiable so that we can construct an asymptotically linear estimator of $\eta_0$. Condition~\ref{cond: derivative} requires that the function $U$ is differentiable so that the delta method can be used to obtain the influence function of $V(f_{P, -S_P}, P)$. These conditions are related to but slightly different from conditions~(A1)--(A4) of \cite{williamson2022general}. In particular, our conditions~\ref{cond: local continuity} and~\ref{cond: derivative} imply condition~(A2) of \cite{williamson2022general}. 

We now show that conditions~\ref{cond: local continuity}--\ref{cond: derivative} hold and provide $\dot{\zeta}$ and $\phi_0$ for the three examples introduced in Section~\ref{sec:vimp}. The proof of these results is also provided in Supplementary Material.
\begin{example}[continues=example:rsquared]
    For the $R$-squared predictiveness metric:  condition~\ref{cond: local continuity} holds with $\dot{\zeta}(f)(x, y) = [y-f(x)]^2$ and $\|\cdot\|_{\s{F}} = \| \cdot\|_{L_{p}(P_0)}$ if $\int |y - f(x)|^q \sd P_0(x,y) < \infty$ for all $f \in \s{F}_{-S}$ in a neighborhood of $f_{0,-S}$, where $p \in [2, \infty)$ and $q \in (2,\infty]$ are such that $2/p + 2/q = 1$; condition~\ref{cond: nuisance} holds with $\phi_0(x, y) = (y - E_0 [Y])^2 - \sigma_0^2$; and condition~\ref{cond: derivative} holds if $\sigma_0^2 > 0$.
\end{example}
\begin{example}[continues=example:deviance]
    For the deviance predictiveness metric:  condition~\ref{cond: local continuity} holds with $\dot{\zeta}(f)(x, y) = \nu(y, f(x))$ and $\|\cdot\|_{\s{F}} = \| \cdot\|_{L_2(P_0)}$ if there exist $\alpha \in (0, 1/2)$ and $\delta >0$ such that $\|f - f_{0, -S}\|_{L_2(P_0)} < \delta$ implies that $P_0(\alpha \leq f(X) \leq 1-\alpha) = 1$; condition~\ref{cond: nuisance} holds with $\phi_0(x, y) = [I\{y = 1\} - \pi_0] \log\left(\frac{\pi_0}{1-\pi_0}\right)$; and  condition~\ref{cond: derivative} holds if $\pi_0 \in (0,1)$.
\end{example}
\begin{example}[continues=example:classification]
    For the classification accuracy predictiveness metric: condition~\ref{cond: local continuity} holds with $\dot{\zeta}(f)(x, y) = I\{y = f(x)\}$ and $\|\cdot\|_{\s{F}} = \| \cdot\|_{\infty}$ if $\int |\mu_0(x) - 1/2|^{-1}  \sd P_0(x) < \infty$; condition~\ref{cond: nuisance} holds with $\phi_0 = 0$; and condition~\ref{cond: derivative} holds.
\end{example}

\subsection{Estimating the proposed measures}\label{sec: estimation}

We now propose estimators of our population parameters of interest. We recall that both of our parameters of interest involve $\psi_{0,S} := V(f_0, P_0) - V(f_{0, -S}, P_0) = U( \zeta( f_{0}, P_0), \eta_0) - U( \zeta( f_{0,-S}, P_0), \eta_0)$ for sets $S \subseteq \{1, \dotsc, p\}$. To estimate $V(f_0, P_0)$ and  $V(f_{0, -S}, P_0)$, we will consider the plug-in estimators $V_{n} := U( \zeta( f_{n}, \d{P}_n), \eta_n)$ and $V_{n,-S} := U( \zeta( f_{n,-S}, \d{P}_n), \eta_n)$, where $f_{n}$ and $f_{n, -S}$ are estimators of $f_0$ and  $f_{0, -S}$, respectively, and $\eta_n$ is an estimator of $\eta_0$. As in \cite{williamson2022general}, $f_{n, -S}$ may be based on data-adaptive nonparametric or semiparametric regression methods. We then define an estimator of $\psi_{0, S}$ as $\psi_{n, S} := V_{n}- V_{n,-S}$.  For our proposed measure $(|S_0|, \psi_{0,S_0})$ of a variable selection algorithm whose selected subset $S_n$ is converging to $S_0$, we then propose the estimator $(|S_n|, \psi_{n,S_n})$. Our estimator differs from that of \cite{williamson2022general} because the selected subset $S_n$ is random.

For our proposed VROC measure $(\psi_{0, R_{0, [1]}}, \dotsc, \psi_{0, R_{0, [p]}})$ of a variable ranking algorithm whose ranks $R_{n}$ (which we recall is a random permutation of $\{1, \dotsc, p\}$) are converging to $R_{0}$, we propose the VROC estimator $(\psi_{n, R_{n, [1]}}, \dotsc, \psi_{n,R_{n, [p]}})$. For the AUVROC $\phi(\psi_{0, R_{0, [1]}}, \dotsc, \psi_{0,R_{0, [p]}})$, we propose the analogous  AUVROC estimator 
\[
    \phi(\psi_{n, R_{n, [1]}}, \dotsc, \psi_{n, R_{n, [p]}}) := \sum_{j=2}^{p}\frac{\psi_{n, R_{n, [j]}} + \psi_{n, R_{n, [j-1]}}}{2}.
\]

\subsection{Asymptotically stable selection and ranking algorithms}\label{sec:stable ranking algorithms}

When defining our measure of variable selection and ranking algorithms, we assumed that the algorithm  converges in probability to a fixed limit. This is often too strong to hold in practice, so we now relax this assumption. One major reason that a ranking or selection algorithm may not converge to a fixed limit is the presence of null variables---that is, variables that don't change the value of predictiveness function. For example, suppose there are $p = 3$ independent variables, but the outcome only depends on the first variable. In this case, there is no true signal upon which to form a ranking of the last two variables, so these variables will be ranked on noise alone. As a result, the ranking $R_n$ may not converge to a single fixed ranking as $n$ increases, as small changes in the noise may change the ranking of the null variables. However, as long as the ranking algorithm asymptotically ranks the first variable first, it is reasonable to expect that $R_n$ will asymptotically be contained in the \emph{set} of rankings $\mathcal{R}_0$ with the first variable ranked first, i.e.,  $\s{R}_0 = \{(1,2,3), (1,3, 2)\}$. Furthermore, the true oracle predictiveness function sequence $(f_{0, -R_{0, [1]}}, f_{0, -R_{0, [2]}}, f_{0, -R_{0, [3]}})$ is the same for each rank $R \in \mathcal{R}_0$ because both $X_2$ and $X_3$ are null variables. Similar situations can also happen for variable selection algorithms. For instance, if an algorithm is designed to select at least two variables, but as above there is only one true non-null variable, the algorithm may not converge to a fixed $S_0$ because it may add null variables to the selected set at random.  However, it may still be reasonable to expect that the algorithm is asymptotically contained in the set $\mathcal{S}_0$ with the first variable always selected, i.e.,  $\s{S}_0 = \{\{1, 2\}, \{1, 3\}\}$. Furthermore, the true oracle predictiveness function is the same for each $S \in \s{S}_0$. This leads us to the relaxed notion of \emph{asymptotically stable} selection and ranking algorithms. 

\begin{define}[Asymptotically stable selection and ranking algorithms]
    A variable selection algorithm $S_n$ is asymptotically stable with limiting selection set $\s{S}_0$ if $P_0(S_n \in \s{S}_0) \to 1$ and $f_{0, -S} = f_{0, -S'} $ for all $S, S' \in \s{S}_0$. A variable ranking algorithm $R_n$ is asymptotically stable with limiting rank set $\s{R}_0$ if $P_0(R_n \in \s{R}_0) \to 1$ and $f_{0, -R_{[j]}} = f_{0, -R_{[j]}'}$ for all $R, R' \in \s{R}_0$ and $j \in \{1, \dotsc, p\}$. 
\end{define}

A selection algorithm is asymptotically stable if it is asymptotically contained in a set of selections with common predictiveness function under $P_0$, and a ranking algorithm is asymptotically stable if it is asymptotically contained in a set of rankings with a common sequence of predictiveness functions under $P_0$. These definitions permit in particular that a selection algorithm includes random, non-converging sets of null variables in the selected set, and that a ranking algorithm ranks null variables in an arbitrary manner. Under this relaxed condition, our parameters of interest are still well-defined, and we will still be able to establish asymptotic results for our estimators. 

It is important to note that there are situations where selection or ranking algorithms may not be asymptotically stable. For example, suppose as above that there are $p = 3$ independent covariates, and that the true response model is $Y = X_1 + X_2^2 + \varepsilon$, for independent noise $\varepsilon$. If $X_2$ has a distribution symmetric around 0 and we rank the variables using MR defined in Section~\ref{sec:illust}, $X_2$ and $X_3$ will both be regarded as null by the model, and hence once again the ranking may not converge to a fixed rank, but instead only be asymptotically contained in the set $\mathcal{R}_0$ defined above. However, unlike the previous example, the predictiveness function sequences of the two ranks in $\mathcal{R}_0$ are no longer the same because $X_2$ is not truly null. Similar examples can be constructed for selection algorithms. Situations like this may not be covered by our theoretical results.

\subsection{Asymptotic linearity}\label{sec: asym linear}

We now establish conditions under which the estimators $V_{n, -S_n}$ and $V_{n, -R_{n, [j]}}$ are asymptotically linear and nonparametric efficient. We first introduce additional conditions, which are again specific to a subset $S \subseteq \{1, \dotsc, p\}$.

\begin{enumerate}[label=\textbf{(B\arabic*)},leftmargin=2cm]
    \item[\namedlabel{cond: convergence rate}{(B1)}] It holds that $\| f_{n, -S} - f_{0, -S} \|_{\s{F}} = o_{\prob_0}(n^{-1/4})$.
    \item[\namedlabel{cond: limited complexity}{(B2)}] There exists a $P_0$-Donsker class $\s{G}$ such that $P_0 (\dot{\zeta}(f_{n, -S}) \in \s{G} ) \to 1$.
    \item[\namedlabel{cond: eta asym linear}{(B3)}] It holds that $\eta_n = \eta_0 + \d{P}_n \phi_0 + o_{\prob_0}(n^{-1/2})$.
\end{enumerate}

\begin{restatable}{thm}{thmasymlinear}\label{thm:asym linearity}
    If the selection algorithm $S_n$ is asymptotically stable with limiting selection set $\s{S}_0$ and conditions~\ref{cond: local continuity}--\ref{cond: derivative} and \ref{cond: convergence rate}--\ref{cond: eta asym linear} hold for all $S \in \s{S}_0$, then for all $S \in \s{S}_0$,
    \[
        V_{n, -S_n}  = V(f_{0, -S}, P_0) + \d{P}_n \dot{V}_0(f_{0, -S}) + o_{\prob_0}(n^{-1/2})
    \]
    with $\dot{V}_0(f_{0, -S})$ given in Theorem~\ref{thm:pathwise differentiable}. If the ranking algorithm $R_n$ is asymptotically stable with limiting ranking set $\s{R}_0$ and conditions~\ref{cond: local continuity}--\ref{cond: derivative} and \ref{cond: convergence rate}--\ref{cond: eta asym linear} hold for all $S=R_{[j]}$ such that $ R \in \s{R}_0$ and $j \in \{1, \dotsc, p\}$, then for all $R \in \s{R}_0$,
    \[
        V_{n, -R_{n, [j]}}  = V(f_{0, -R_{[j]}}, P_0) + \d{P}_n \dot{V}_0(f_{0, -R_{[j]}}) + o_{\prob_0}(n^{-1/2})
    \]
    with $\dot{V}_0(f_{0, -R_{[j]}})$ given in Theorem~\ref{thm:pathwise differentiable}. 
\end{restatable}

Theorem~\ref{thm:asym linearity} provides conditions under which the plug-in estimators  $V_{n, -S_n}$ and $V_{n, -R_{n, [j]}}$ are asymptotically linear with influence functions equal to the nonparametric efficient influence functions established in Theorem~\ref{thm:pathwise differentiable}. Theorem~\ref{thm:asym linearity} will be used to facilitate statistical inference for our parameters of interest in Section~\ref{sec:inference}.

Theorem~\ref{thm:asym linearity} differs from the results of \cite{williamson2022general} in that the set of covariates are random subsets of the $p$ covariates. However, other than the limiting selection set $\s{S}_0$ and limiting ranking set $\s{R}_0$, the variable selection and ranking algorithms do not play a role in the influence functions, and hence the limit distributions, of the estimators. This is because both the selected subset and variable ranks are discrete parameters, and the asymptotic stability assumption ensures that the true predictiveness of the estimated subsets equals a fixed limit asymptotically. However, in finite samples, the behavior of the selection or ranking algorithm can contribute to the sampling distribution of our estimators, which is not captured in the first-order asymptotic results of Theorem~\ref{thm:asym linearity}. This can result in under-coverage of confidence intervals based on Theorem~\ref{thm:asym linearity}. We propose an alternative bootstrap inference procedure in Section~\ref{sec:bootstrap} to address this issue.

The estimators considered in Theorem~\ref{thm:asym linearity} are based on the plug-in principle. Usually, plug-in estimators based on data-adaptive nuisance estimators inherit non-negligible asymptotic bias from the nuisance estimator, which hinders valid statistical inference for the parameter of interest. However, Theorem~\ref{thm:asym linearity} demonstrates that plug-in estimators do not suffer from this problem in this case.  As discussed in \cite{williamson2022general}, this is because $f_{0, -S}$ is a maximizer of $f \mapsto V(f, P_0)$ over $\s{F}_{-S}$, so that we may expect $\left.\frac{\partial}{\partial \varepsilon}  V(f_{\varepsilon, -S}, P_0) \right|_{\varepsilon = 0} = 0$ for a sufficiently smooth path $P_\varepsilon$ through $P_0$ at $\varepsilon = 0$. As a result, we can expect that $V(f, P_0) - V(f_{0,-S}, P_0)$ is bounded locally by $\|f - f_{0, -S}\|_{\s{F}}^2$, as required by condition~\ref{cond: local continuity}. Hence, the plug-in bias $V(f_{n, -S}, P_0) - V(f_{0, -S}, P_0)$ is controlled by the behavior of $\|f_{n, -S} - f_{0, -S}\|_{\s{F}}^2$, so that if $\|f_{n, -S} - f_{0, -S}\|_{\s{F}} = o_{\prob_0}(n^{-1/4})$, as required by condition~\ref{cond: convergence rate}, then the plug-in bias is $o_{\prob_0}(n^{-1/2})$. 

Since the rate $n^{-1/4}$ is slower than $n^{-1/2}$, condition~\ref{cond: convergence rate} can in principle be satisfied by nonparametric and semiparametric estimators, and does not require the use of parametric estimators for $f_{n}$. However, the rate $n^{-1/4}$ is not achievable without some smoothness or structural assumptions about $f_{0}$, and the strength of these assumptions needs to increase with $p$ due to the curse of dimensionality \citep{buhlmann2011statistics}. For example,  the minimax optimal rate of convergence of an estimator of $f_0$ in a model where $f_0$ is assumed to be $m$ times differentiable is  $n^{-m/(2m+p)}$ \citep{stone1982optimal}. Hence, to achieve a rate faster than $n^{-1/4}$, one would need $m > p / 2$. Similarly, if the covariates lie on a $d$-dimensional manifold in $\d{R}^p$ for $d < p$, then the rate $n^{-1/4}$ can be achieved if $f_0$ belongs to a Sobolev class with smoothness $\alpha$ for $\alpha > d/2$  \citep{bickel2007local}. If $f_0$ is assumed to be additive and differentiable, the rate $n^{-1/4}$ can be achieved for any $p$ \citep{stone1985additive}. As a final example, if $f_0$ is known to be a sparse function that depends only on $d \leq \min\{n, p\}$ variables, and $f_0$ belongs to a H{\"o}lder $\alpha$-smooth class, then the rate $n^{-1/4}$ can be achieved if $\alpha > d/2$ and $n^{-1/2} d \log(p/d) \to 0$ \citep{yun2015adaptive}.  However, whether the true function possesses these or other properties, and hence which regression estimator is best suited to the data, is typically unknown in practice. One approach to dealing with this uncertainty is to select between or combine several candidate estimators using cross-validation. For example, SuperLearner \citep{van2007super} is a generalization of the stacking algorithm that combines multiple candidate regression estimators, and achieves at least the best rate of convergence of the candidate estimators \citep{van2011targeted}.

Condition~\ref{cond: limited complexity} requires that the estimated derivative function $\dot{V}_0(f_{n, -S})$ falls in a $P_0$-Donsker class with probability tending to one, which is used to show that the empirical process term satisfies $(\d{P}_n - P_0)(\dot{V}_0(f_{n, -S}) - \dot{V}_0(f_{0, -S})) = o_{\prob_0}(n^{-1/2})$.  Satisfying this condition typically requires  restricting the complexity of the function class $\s{F}$. A main way this is accomplished is by using bracketing or uniform entropy \citep{van1996weak}. However, condition~\ref{cond: limited complexity} can fail when $f_{n, -S}$ is based on highly data-adaptive algorithms, as discussed in \cite{zheng2011cross} and \cite{chernozhukov2018double}. Another approach to controlling the empirical process term is cross-fitting, which we discuss in Section~\ref{sec:cv} below.  Condition~\ref{cond: eta asym linear} requires that the estimator $\eta_n$ of $\eta_0$ is asymptotically linear with the influence function $\phi_0$.

\subsection{Estimators based on cross-fitting}\label{sec:cv}

In this section, we define modified estimators based on cross-fitting and provide alternative asymptotic linearity results for these estimators. As discussed above, cross-fitting removes the need for the Donsker condition~\ref{cond: limited complexity}, which is particularly useful when using data-adaptive function estimators. 

We begin by obtaining the variable selection $S_n$ and rank $R_n$ on the whole data set. Then, we randomly (independently of the data) partition the data into $K \geq 2$ folds with roughly equal sizes. Here, ``roughly equal sizes" means that $\lim_{n \to \infty} \max_{1\leq k \leq K} \abs{\frac{n}{K n_k} - 1} = 0$, where $n_k$ is the size of the $k$-th fold. For simplicity, we assume that $K$ is fixed. For each $k \in \{1, \dotsc, K\}$, we construct estimators $f_{n, k}$, $f_{n, k, -S}$, $f_{n, k, -R_{[j]}}$, and $\eta_{n,k}$ of $f_{0}$, $f_{0, -S}$, $f_{0, -R_{[j]}}$, and $\eta_{0}$, respectively, based on the \emph{training set} for fold $k$---i.e., the data excluding fold $k$. For each $k$, we then define $V_{n, k} := U(\zeta(f_{n,k}, \d{P}_{n,k}), \eta_{n,k})$, $V_{n, k, -S_n} := U(\zeta(f_{n,k, -S_n}, \d{P}_{n,k}), \eta_{n,k})$,  and $V_{n, k, -R_{n,[j]}} := U(\zeta(f_{n,k, -R_{n,[j]}}, \d{P}_{n,k}), \eta_{n,k})$, where $\d{P}_{n,k}$ is the empirical distribution of the $k$th fold. Finally, we construct the $K$-fold cross-fitting estimators of $V(f_0, P_0)$, $V(f_{0,-S}, P_0)$ and $V(f_{0,-R_{[j]}}, P_0)$ as $\frac{1}{K}\sum_{k=1}^{K} V_{n, k}$, $\frac{1}{K}\sum_{k=1}^{K} V_{n, k, -S_n}$, and $\frac{1}{K}\sum_{k=1}^{K} V_{f_{n, k, -R_{n, [j]}}}$,  respectively.

We introduce the following conditions, which as above are specific to a subset $S \subseteq \{1, \dotsc, p\}$.
\begin{enumerate}[label=\textbf{(C\arabic*)},leftmargin=2cm]    
    \item[\namedlabel{cond: cv convergence rate}{(C1)}] It holds that $\| f_{n, k, -S} - f_{0, -S} \|_{\s{F}} = o_{\prob_0}(n^{-1/4})$ for each $k \in \{1, \dotsc, K\}$.
    \item[\namedlabel{cond: cv weak consistency}{(C2)}] It holds that $E_0 \|\dot{\zeta}(f_{n, k, -S}) - \dot{\zeta}(f_{0, -S})\|_{L_2(P_0)} = o(1)$ for each $k \in \{1, \dotsc, K\}$.
    \item[\namedlabel{cond: cv eta asym linear}{(C3)}] It holds that $\eta_{n, k} = \eta_0 + \d{P}_{n, k} \phi_0 + o_{\prob_0}({n_k}^{-1/2})$.
\end{enumerate}

\begin{restatable}{thm}{thmcvclassical}\label{thm: cv asym linear}
    If the selection algorithm $S_n$ is asymptotically stable with limiting selection set $\s{S}_0$ and conditions~\ref{cond: local continuity}--\ref{cond: derivative} and \ref{cond: cv convergence rate}--\ref{cond: cv eta asym linear} hold for $S \in \s{S}_0$,  then
    \[
        \frac{1}{K}\sum_{k=1}^{K} V_{n, k, -S_n}  = V(f_{0, -S}, P_0) + \d{P}_n \dot{V}_0(f_{0, -S}) + o_{\prob_0}(n^{-1/2}).
    \]
    If the ranking algorithm $R_n$ is asymptotically stable with limiting ranking set $\s{R}_0$ and conditions~\ref{cond: local continuity}--\ref{cond: derivative} and \ref{cond: cv convergence rate}--\ref{cond: cv eta asym linear} hold for all $S=R_{[j]}$ where $ R \in \s{R}_0$ and $j \in \{1, \dotsc, p\}$, then
    \begin{equation*}
        \frac{1}{K}\sum_{k=1}^{K} V_{n, k, -R_{n, [j]}}  = V(f_{0, -R_{[j]}}, P_0) + \d{P}_n \dot{V}_0(f_{0, -R_{[j]}}) + o_{\prob_0}(n^{-1/2}).
    \end{equation*}  
\end{restatable}

Theorem~\ref{thm: cv asym linear} provides conditions under which the cross-fitting estimators are asymptotically linear with influence functions equal to the nonparametric efficient influence functions established in Theorem~\ref{thm:pathwise differentiable}. Notably, asymptotic linearity of the cross-fitting estimators does not require the Donsker condition~\ref{cond: limited complexity}. Condition~\ref{cond: cv convergence rate} requires that the maximal error of the nuisance estimators converges faster than $n^{-1/4}$, which is analogous to condition~\ref{cond: convergence rate} used to establish Theorem~\ref{thm:asym linearity}. Condition~\ref{cond: cv weak consistency} requires convergence in mean of the estimated influence function, which is used to control the empirical process term of the cross-fitting estimator. To control the empirical process term of the non-cross-fitting estimators in Theorem~\ref{thm:asym linearity}, we only needed that $\|\dot{\zeta}(f_{n, -S}) - \dot{\zeta}(f_{0, -S})\|_{L_2(P_0)} = o_{\prob_0}(1)$, which follows from conditions~\ref{cond: local continuity} and~\ref{cond: convergence rate}. Hence, we did not require an analogue of condition~\ref{cond: cv weak consistency} in the conditions for Theorem~\ref{thm:asym linearity}. However, condition~\ref{cond: cv weak consistency} would follow from conditions~\ref{cond: local continuity} and~\ref{cond: cv convergence rate} if $\sup_{f \in \s{F}}|\dot{\zeta}(f)|$ is uniformly bounded. In Section~\ref{sec:sim}, we will provide simulations comparing the properties of the estimator with and without cross-fitting. Condition~\ref{cond: cv eta asym linear} requires that the estimator $\eta_{n, k}$ of $\eta_0$ is asymptotically linear with influence function $\phi_0$ for each fold.

\section{Large-sample statistical inference}\label{sec:inference}\label{sec:bootstrap}

Theorems~\ref{thm:asym linearity} and~\ref{thm: cv asym linear} can be used to construct asymptotically valid confidence intervals for $\psi_{0, S}$, where $S \in \s{S}_0$, as long as $\psi_{0,S} > 0$. If $\sigma_n^2$ is a consistent estimator of  $\sigma_0^2 :=  P_0\left[\dot{V}_0(f_0) - \dot{V}_0(f_{0, -S})\right]^2$, then a Wald confidence interval for $\psi_{0, S}$ is given by $\psi_{n, S_n} \pm z_{1-\alpha/2}n^{-1/2}\sigma_n$, where $z_{p}$ is the lower $p$th quantile of the standard normal distribution. Theorems~\ref{thm:asym linearity} and~\ref{thm: cv asym linear} can also be used to construct asymptotically valid confidence intervals for each $\psi_{0, R_{[j]}}$, as long as $\psi_{0,R_{[j]}} > 0$, as well as uniformly valid confidence sets for  $(\psi_{0, R_{[1]}}, \dotsc, \psi_{0, R_{[p]}})$, where $R \in \s{R}_0$, as long as $\psi_{0, R_{[1]}} > 0$. The asymptotic variance can be estimated using so-called \emph{influence function-based} estimators. Since these methods are well-known, we omit the details here, but provide them in Supplementary Material.

The bootstrap is an alternative method of constructing confidence intervals  \citep{efron1982jackknife, efron1994introduction}. In some settings, bootstrap confidence intervals have been shown to have higher-order accuracy and better finite-sample coverage than Wald intervals \citep{diciccio1988review, hall1988theoretical, hall1992bootstrap}. In our setting, the bootstrap may be able to address at least two sources of potential finite-sample bias in the large-sample confidence intervals defined above. First, even if the selection or ranking algorithm is asymptotically stable, it may possess variability in finite samples that is not captured by the influence function-based variance estimators. Second, while the precise behavior of the prediction estimators does not play a role in the asymptotic distribution of our estimators as long as the prediction estimators satisfy the rate and complexity conditions, they may contribute to the finite-sample variability of our estimators. Accounting for these two sources of additional variability could improve the properties of our confidence intervals.

To implement a standard empirical bootstrap, we would generate $n$ IID samples from the empirical distribution $\d{P}_n$ and use the bootstrap data to construct a bootstrap estimator in the exact same manner as the estimator was constructed using the original data. 
However, this standard approach has two shortcomings for our estimators. First, to avoid model misspecification, we advocate for using nonparametric or semiparametric regression estimators or even an ensemble of many such estimators to estimate the prediction functions $f_{0, -S}$ and $f_{0, -R_{[j]}}$. Such estimation procedures may be computationally intensive, and repeating this computationally intensive procedure for every bootstrap sample may be infeasible since the number of bootstrap samples $B$ is typically in the hundreds or thousands. Second, the bootstrap can fail if the estimator is sensitive to replicated observations \citep{bickel1997resampling}, which may be the case for our estimators. Many data-adaptive regression estimators, including ensemble estimators such as SuperLearner \citep{van2007super}, involve cross-validation as part of the procedure. When there are replicated observations in the data,  the same observation can appear in the training and test sets, which breaks the independence of training and test sets and leads to overfitting \citep[Chapter~28]{van2018targeted}. 

To address these two problems with the standard empirical bootstrap, we propose a modified \textit{partial bootstrap} procedure. Specifically, when constructing our estimators using the bootstrap data, we use the prediction function estimator based on the \emph{original data} rather than constructing new prediction function estimators using the bootstrap data. This greatly reduces the computational burden of the bootstrap, and in some cases makes the bootstrap feasible when it would not have been. In addition, issues with replicated observations due to the prediction function estimator are also resolved. We note that by fixing the prediction estimator, our proposed partial bootstrap does not account for variability in this estimator, which may result in worse finite-sample performance than a procedure that does account for this variability. However, our partial bootstrap does account for variability in the selection or ranking algorithm. We note that since we are proposing to bootstrap the variable selection or ranking algorithm, we are implicitly assuming this algorithm is not sensitive to replicated observations.

Here are the details of our partial bootstrap procedure for the $K$-fold estimator based on cross-fitting. As in Section~\ref{sec:cv}, we randomly partition the indices $\{1, \dotsc, n\}$ into $K$ folds $I_1, \dotsc, I_{K}$ of roughly equal sizes, and for each $k$ we construct the prediction function estimator $f_{n, k, -S}$ using the data excluding the indices in $I_k$. For each $b \in \{1, \dotsc, B\}$, we construct the $b$th empirical bootstrap estimator as follows. As usual for the empirical bootstrap, we first draw $n$ indices $(\alpha_{b,1}, \dotsc, \alpha_{b,n})$  IID from a uniform distribution on $\{1, \dotsc, n\}$, and we define $\left\{\left(X_{\alpha_{b,i}}, Y_{\alpha_{b,i}}\right): i = 1, \dotsc, n\right\}$ as the bootstrap data (which typically contains replicates of the original observations). We then estimate the variable selection $S_{n}^{(b)}$ and variable rank $R_{n}^{(b)}$ using the bootstrap data. Next, for each $k \in \{1, \dotsc, K\}$, we define $\d{P}_{n,k}^{(b)}$ as the empirical distribution of the bootstrap data whose indices fall in $I_k$. We then define $\psi_{n, k, S}^{(b)} := V(f_{n,k}, \d{P}_{n,k}^{(b)}) -V(f_{n, k, -S}, \d{P}_{n,k}^{(b)})$ and $\psi_{n,S}^{(b)} := \sum_{k=1}^{K}  w_{n,k}^{(b)} \psi_{n, k, S}^{(b)}$, where $w_{n,k}^{(b)} := \frac{1}{n} \sum_{i=1}^n I(\alpha_{b,i} \in I_k)$. Finally, we use the bootstrap estimates $\left(\psi_{n,S}^{(1)}, \dotsc, \psi_{n,S}^{(B)}\right)$ to construct bootstrap confidence intervals. 

\section{Numerical studies}\label{sec:sim}

We conducted a simulation study to validate the large-sample results of Section~\ref{sec:theory} and to assess the finite-sample performance of the proposed methods. We considered the following data-generating process $P_0$. We first  generated  $(X_1, X_2, X_3, X_4, X_5)$ from a multivariate normal distribution with $E_0[X_j] = 0$ and $\mathrm{Var}_0(X_j)= 1$ for all $j$ and $\mathrm{Cor}_0(X_j, X_k) = 0.4$ for each $j \neq k$. Given $(X_1, \dotsc, X_5)$, we then generated $Y$ from a normal distribution with mean 
\[
    \mu_0(x_1, \dotsc, x_5) := (x_1+0.5)(x_2+1)  + \sqrt{\max(x_2 + 5, 0)} + 5\sqrt{(x_3-0.2)^2 + 1}
\]
and variance $1 + \abs{x_4} + \abs{x_5}$. Hence, our data-generating process involved an interaction term, non-linear terms, null variables, correlated variables, and heteroskedasticity. 

We considered three ranking algorithms. First, we considered the ranking obtained from the coefficient path from a LASSO regression \citep{tibshirani1996regression} using the default settings in the \texttt{glmnet} package in \texttt{R}. Second, we considered the ranking obtained by ordering the $p$-values from smallest to largest from a generalized additive model \citep{hastie1990generalized} using the default settings in the \texttt{mgcv} package in \texttt{R}. Finally, we considered the ranking obtained by ordering the estimated variable importances from smallest to largest from a multivariate adaptive regression splines regression \citep{friedman1991multivariate} using the default settings in the \texttt{earth} package in \texttt{R}. Throughout this section, we abbreviate ``generalized additive model" as ``GAM" and ``multivariate adaptive regression splines" as ``MARS". By simulating a large number of samples and inspecting the rankings for each procedure,  we determined that MARS and GAM were converging to the limiting rank set $\s{R}_0 = \{(3,1,2,4,5), (3,1,2,5,4)\}$ with true $R$-squared predictiveness sequence $(0.39, 0.55, 0.66, 0.70, 0.75)$ and AUVROC equal to $2.48$, while LASSO was converging to the limiting rank set $\s{R}_0 = \{(1,2,3,4,5), (1,2,3,5,4)\}$ with true $R$-squared predictiveness sequence $(0.14, 0.23, 0.66, 0.70, 0.75)$ and AUVROC 2.03.

We considered three selection algorithms. First, we considered the subset obtained from a LASSO regression using penalty parameter selected by ten-fold cross-validation and default settings in the \texttt{glmnet} package in \texttt{R}. Second, we considered the selection obtained by including the variables that have $p$-value smaller than 0.05 from a GAM using the default settings in the \texttt{mgcv} package in \texttt{R}. Finally, we considered the selection obtained from a MARS regression using the default settings in the \texttt{earth} package in \texttt{R}. By simulating a large number of samples and inspecting the selection for each procedure,  we determined that LASSO and MARS were converging to the limiting selection set $\s{S}_0 = \{\{1,2,3\}\}$ with true $R$-squared predictiveness measure $0.66$ and PPSV $0.22$, and GAM was converging to the limiting selection set $\s{S}_0 = \{\{1,2,3,4,5\}\}$  with true $R$-squared predictiveness measure $0.75$ and PPSV $0.15$. 

For each sample size $n$ equal to 500, 1000, 2000, 3000, 4000, 5000, 10000 and 20000, we simulated 500 samples of $n$ IID observations from the above data-generating mechanism. We considered $V$ equal to the $R$-squared predictiveness metric defined in Section~\ref{sec:vimp}. We estimated our proposed measures $\psi_{0,S_0} / |S_0|$ and AUVROC using both the non-cross-fitting and cross-fitting estimators with $K = 5$ folds. To estimate the regression functions, we used SuperLearner \citep{van2007super} with five-fold cross-validation and a library consisting of \texttt{xgboost} \citep{chen2016xgboost, xgboostpackage}, \texttt{gam} \citep{hastie1990generalized, gampackage}, and \texttt{earth} \citep{friedman1991multivariate, earthpackage}. We constructed 95\% Wald intervals using the cross-fitting influence function-based variance estimator as in Section~\ref{sec:inference}. We also used the partial bootstrap procedure defined in \Cref{sec:bootstrap} to construct bootstrap confidence intervals. We considered three types of bootstrap confidence intervals: the percentile method, percentile $t$-method, and Efron's percentile method (in the terminology of \citealp{van2000asymptotic}).

Figure~\ref{fig:estimate} displays the true VROC curves for the three ranking algorithms along with cross-fitting estimators and pointwise Wald and uniform 95\% confidence sets for a single simulation with $n=1000$. GAM and MARS both produced the ranking $(3, 1, 2, 4, 5)$, and LASSO produced the ranking $(1, 2, 3, 4, 5)$. The results for GAM and MARS are the same in Figure~\ref{fig:estimate} because they produced the same ranking. The estimated AUVROC for GAM and MARS was $2.55$ (95\% CI: 2.40--2.69), and the estimated AUVROC for LASSO was $2.07$ (1.96--2.18). Using the variable selection algorithms described above, GAM selected $\{X_1, X_2, X_3\}$, LASSO selected $\{X_1,X_2,X_3,X_5\}$ and MARS selected $\{X_1,X_2,X_3,X_4,X_5\}$. In this case, the estimated PPSV for GAM, LASSO and MARS were 0.23 (0.21--0.24), 0.17 (0.16--0.18), and 0.15 (0.14--0.16) respectively.

\begin{figure}[!htbp]
    \centering
    \includegraphics[scale=0.7]{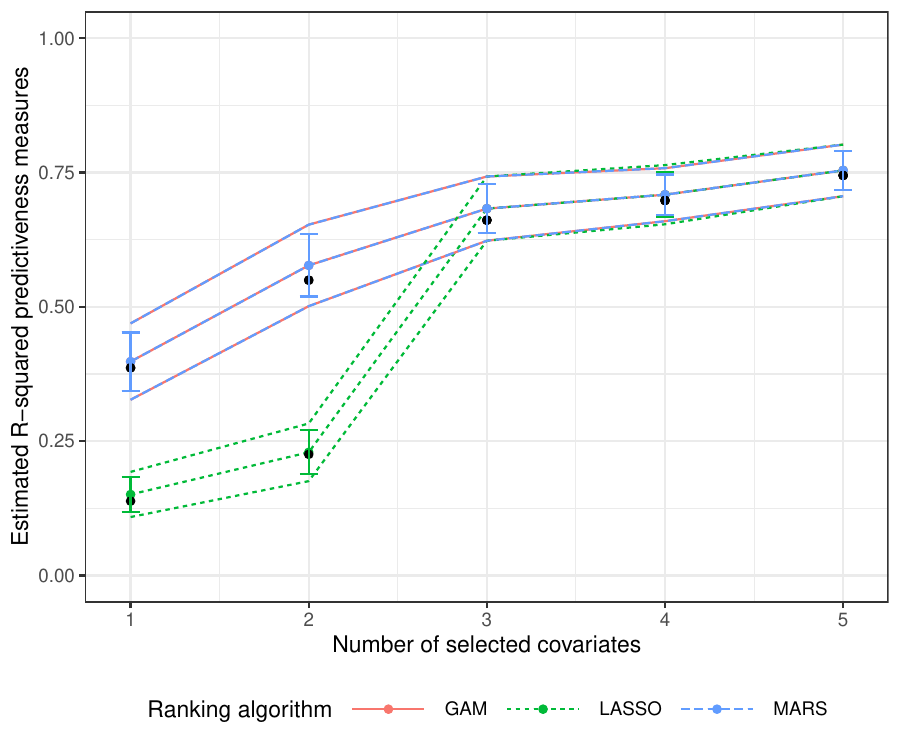}
    \caption{
    True VROC (black points) for the three different variable ranking algorithms considered along with estimated VROC curves and 95\% pointwise intervals displayed as error bars and 95\% uniform confidence regions displayed as lines for a single simulation with $n=1000$.}
    \label{fig:estimate}
\end{figure}

We now turn to the results of the simulation study. Figure~\ref{fig:asymptotic behavior of estimators} displays the properties of the AUVROC (top row) and PPSV (bottom row) estimators and corresponding Wald confidence intervals. The left column displays $n^{1/2}$ times the bias of the estimators. In all cases, the bias appears to tend to zero faster than $n^{-1/2}$ for large enough sample sizes, but there is considerable heterogeneity in the finite-sample bias of the estimators. The absolute bias of the estimators based on cross-fitting is generally smaller than the bias of the corresponding estimators without cross-fitting. The bias of the AUVROC estimators without cross-fitting for the GAM and MARS ranking algorithms in particular is substantial. This demonstrates the importance of cross-fitting in reducing bias when data-adaptive nuisance estimators, such as extreme gradient boosting, are used to estimate the predictiveness function. For the PPSV, the absolute bias appears to decrease slower than $n^{-1/2}$ for $n$ less than roughly 5000. As we discuss more below, this is because the sampling distribution of the PPSV estimators is multi-modal at these sample sizes.

The middle column of Figure~\ref{fig:asymptotic behavior of estimators} displays $n^{1/2}$ times the standard deviation of the estimators. The standard deviations appear to stabilize at the $n^{-1/2}$ rate for all estimators except that of the PPSV with the GAM algorithm. This is because the subset selected by the GAM algorithm is still not stable at sample size $20000$, and the standard deviation decreases as the selection algorithm stabilizes. The standard deviation of the cross-fitting estimators is larger than that of the non-cross-fitting estimators for sample sizes less than 5000 because the cross-fitting estimators use a smaller training set when estimating the predictiveness function.

The right column of Figure~\ref{fig:asymptotic behavior of estimators} displays the empirical coverage rate of  95\% Wald confidence intervals for the estimators. The excess bias for the non-cross-fitting estimator translates to worse confidence interval coverage for both the AUVROC and PPSV. The confidence intervals for the AUVROC based on the cross-fitting estimators are valid in large sample sizes and have generally good performance for $n \geq 2000$. However, for the PPSV, even the confidence intervals based on cross-fitting have poor coverage. This is because the Wald confidence intervals do not take the variability of the selected set $S_n$ into account. The denominator of the PPSV estimator is $|S_n|$, which is integer-valued and hence varies substantially unless the selection algorithm $S_n$ is very stable. This results in a multi-modal sampling distribution of the PPSV estimator. The Wald interval uses a normal approximation to the sampling distribution, but the normal distribution is a bad approximation to the true multi-modal sampling distribution. The multi-modality of the sampling distribution is shown in additional figures in Supplementary Material. By sample size 20000, the selected set $S_n$ is stable enough for the MARS and LASSO algorithms that the Wald confidence intervals based on cross-fitting have close to nominal coverage. However, the set selected by GAM is not stable even at this sample size, which results in poor coverage. As we will discuss below, our partial bootstrap procedure can in some cases improve the coverage of the PPSV.

\begin{figure}[t!]
    \centering
    \includegraphics[width=\linewidth]{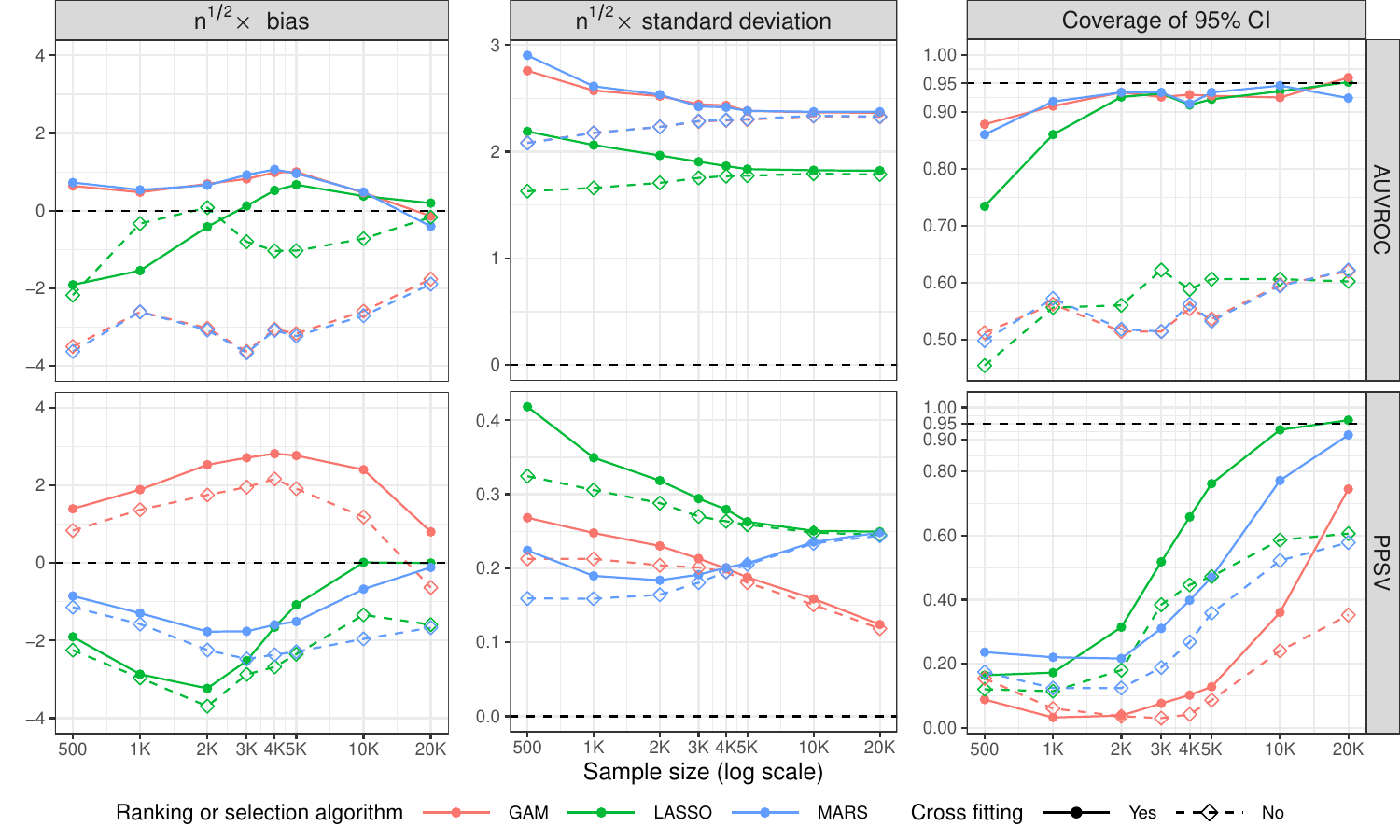}
    \caption{Properties of the AUVROC (top row) and PPSV (bottom row) estimators with and without cross-fitting for the three variable ranking and selection algorithms.}
    \label{fig:asymptotic behavior of estimators}
\end{figure}

Figure~\ref{fig:asymptotic behavior of bootstrap estimates} displays the empirical coverage rate of 95\% partial bootstrap confidence intervals based on the cross-fitting estimators for the three algorithms.  All bootstrap confidence intervals have close to 95\% coverage for sample size 20000. The coverage of all three types of bootstrap intervals for the AUVROC (top row) have comparable or better coverage than the Wald interval for small and moderate samples, and have coverage greater than 90\% for all cases when $n \geq 1000$.  The partial bootstrap yields better coverage than the Wald intervals for the AUVROC because the bootstrap procedure incorporates the variability of the ranking algorithm. For the PPSV (bottom row), the confidence intervals again have better coverage than the Wald intervals. This is again because the bootstrap procedure captures the variability of the  selection algorithm. However, in this case, the percentile and percentile-t methods still have far from nominal coverage for $n < 10000$. Efron's percentile method, which uses the quantiles of the sampling distribution of the estimator directly to construct confidence intervals, has the best coverage by a wide margin. For the GAM algorithm, Efron's percentile method has close to nominal coverage for most sample sizes, though it still requires larger sample sizes for the LASSO and MARS algorithms. We hypothesize that Efron's percentile method performs better than the other methods in this case because it uses the sampling distribution of the uncentered bootstrap estimator $\psi_{n, S_n}^{(j)} / |S_n^{(j)}|$ given the data to approximate the sampling distribution of  $\psi_{n, S_n} / |S_n|$, while the percentile and percentile-$t$ use the distribution of the centered bootstrap estimator $\psi_{n, S_n}^{(j)} / |S_n^{(j)}| - \psi_{n, S_n} / |S_n|$ given the data to approximate the sampling distribution of $\psi_{n,S_n} / |S_n| - \psi_{0, S} / |S|$. The uncentered bootstrap distribution may be a better approximation to the multi-modality of the true sampling distribution of $\psi_{n, S_n} / |S_n|$ than the distribution of the centered bootstrap estimator is to the true sampling distribution of $\psi_{n, S_n} / |S_n| - \psi_{0, S} / |S|$. 

\begin{figure}[t!]
    \centering
    \includegraphics[width=\linewidth]{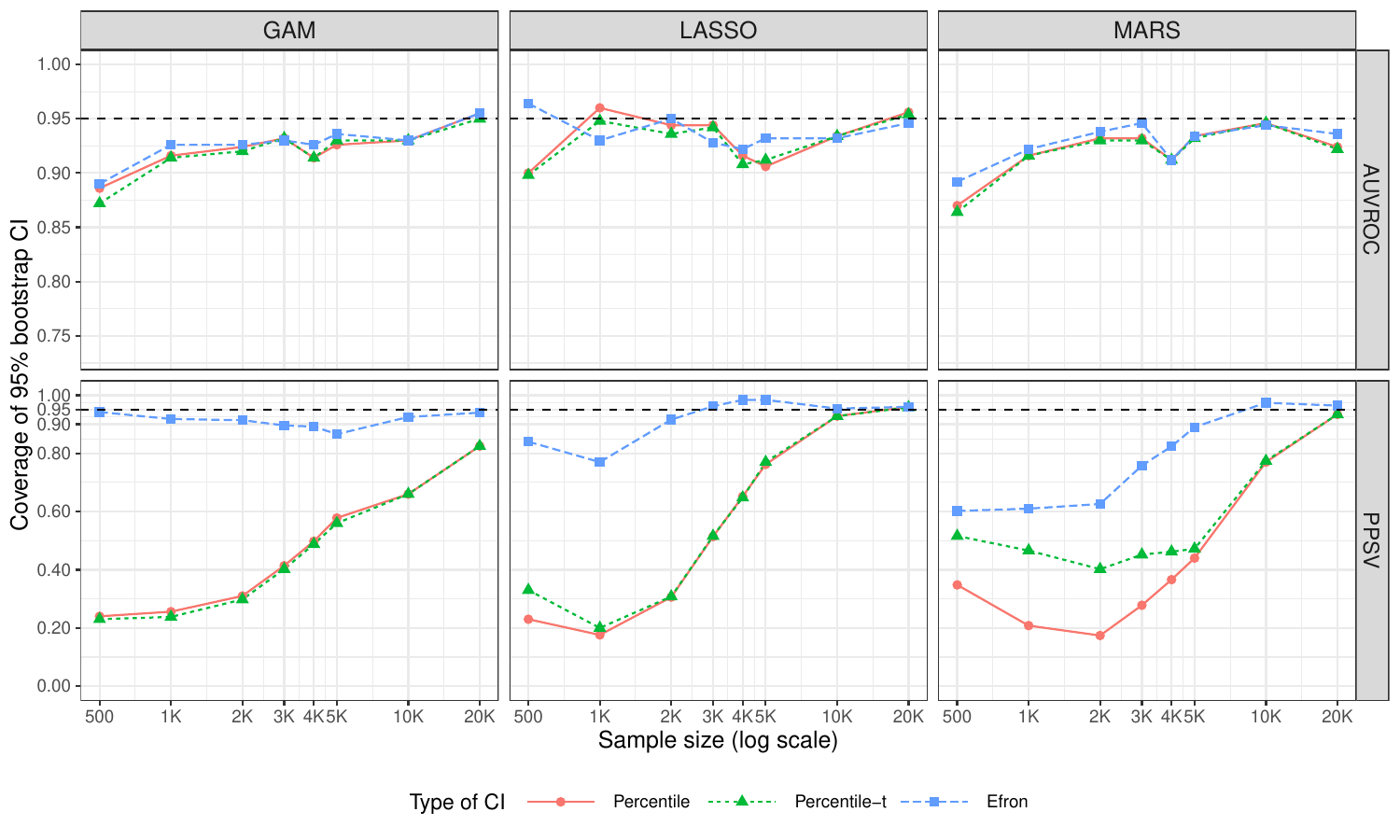}
    \caption{Empirical coverage of 95\% confidence intervals using the partial bootstrap for the AUVROC (top row) and PPSV (bottom row) with cross-fitting for the three ranking and selection algorithms. Note as \Cref{fig:asymptotic behavior of estimators}.}
    \label{fig:asymptotic behavior of bootstrap estimates}
\end{figure}

\section{Comparison of variable selection and ranking algorithms for predicting wine quality}\label{sec:analysis}

In this section, we use the methods developed in this article to compare variable selection and ranking algorithms for predicting the quality of wine. We use the data described in \cite{paulo2009modeling}, which are publicly available at \url{https://archive.ics.uci.edu/ml/datasets/Wine+Quality}. This data contains 11 physicochemical properties of $n=4898$ different \textit{vinho verde} wines. We treat these as our covariates $X$. The data also contain a quality score between 0 and 10, which was computed as the median of at least three blind taste tests. We treat this as the outcome $Y$. Hence, our goal is to predict the subjective rating of a wine based on its physicochemical properties.

We used LASSO, GAM, and MARS to rank and select among the 11 physicochemical properties in terms of their importance for predicting wine quality. We refer the reader to Section~\ref{sec:sim} for precise explanations of these ranking and selection algorithms. We evaluated these algorithms using the $R$-squared predictiveness metric. As in Section~\ref{sec:sim}, to estimate the regression function, we used SuperLearner with candidate library consisted of \texttt{xgboost}, \texttt{gam}, and \texttt{earth}. We used the estimators based on cross-fitting with $K = 5$ folds defined in Section~\ref{sec:cv}, and we constructed pointwise confidence intervals using the Wald method and equi-precision uniform confidence bands with influence function-based (co)variance estimator.

\begin{table}[!htbp]
    \centering
    \begin{tabular}{c|
    >{\centering}m{0.8cm}
    >{\centering}m{0.8cm}
    >{\centering}m{0.8cm}
    >{\centering}m{0.8cm}
    >{\centering}m{0.8cm}
    >{\centering}m{0.8cm}
    >{\centering}m{0.8cm}
    >{\centering}m{0.8cm}
    >{\centering}m{0.8cm}
    >{\centering}m{0.8cm}
    >{\centering\arraybackslash}m{0.8cm}
    }
        \toprule
         & \rotatebox[origin=c]{90}{\parbox[c]{2cm}{\centering volatile acidity}} & \rotatebox[origin=c]{90}{\parbox[c]{2cm}{\centering free sulfur dioxide}} & \rotatebox[origin=c]{90}{\parbox[c]{2cm}{\centering residual sugar}} & \rotatebox[origin=c]{90}{\parbox[c]{2cm}{\centering alcohol}} & \rotatebox[origin=c]{90}{\parbox[c]{2cm}{\centering citric acid}} & \rotatebox[origin=c]{90}{\parbox[c]{2cm}{\centering fixed acidity}}  & \rotatebox[origin=c]{90}{\parbox[c]{2cm}{\centering density}}  & \rotatebox[origin=c]{90}{\parbox[c]{2cm}{\centering pH}} & \rotatebox[origin=c]{90}{\parbox[c]{2cm}{\centering sulphates}} & \rotatebox[origin=c]{90}{\parbox[c]{2cm}{\centering total sulfur dioxide}}  & \rotatebox[origin=c]{90}{\parbox[c]{2cm}{\centering chlorides}} \\ 
        \midrule
        MARS & 1 & 2 & 3 & 4 & 5 & 6 & 7 & 8 & 9 & 10 & 11 \\ 
        GAM & 1 & 2 & 3 & 8 & 6 & 9 & 4 & 5 & 7 & 10 & 11 \\ 
        LASSO & 2 & 3 & 4 & 1 & 11 & 5 & 9 & 8 & 7 & 10 & 6\\
        \bottomrule
    \end{tabular}
    \caption{Rank of 11 physicochemical properties for predicting wine rating using the three different ranking algorithms: MARS, GAM, and LASSO. The variables were sorted by the rank of MARS.}
    \label{tab: all rank}
\end{table}

\Cref{tab: all rank} displays the ranking of the physicochemical properties returned by the three ranking algorithms. All three algorithms ranked volatile acidity, free sulfur dioxide, and residual sugar among the top four most important variables, and MARS and GAM agreed on the ordering of these three. MARS and LASSO agree that alcohol is also important, but the rankings diverge somewhat after these first four rankings. All three algorithms ranked total sulfur dioxide second to last, and MARS and GAM both ranked chlorides last.

\begin{figure}[!htbp]
    \centering
    \includegraphics[width=\linewidth]{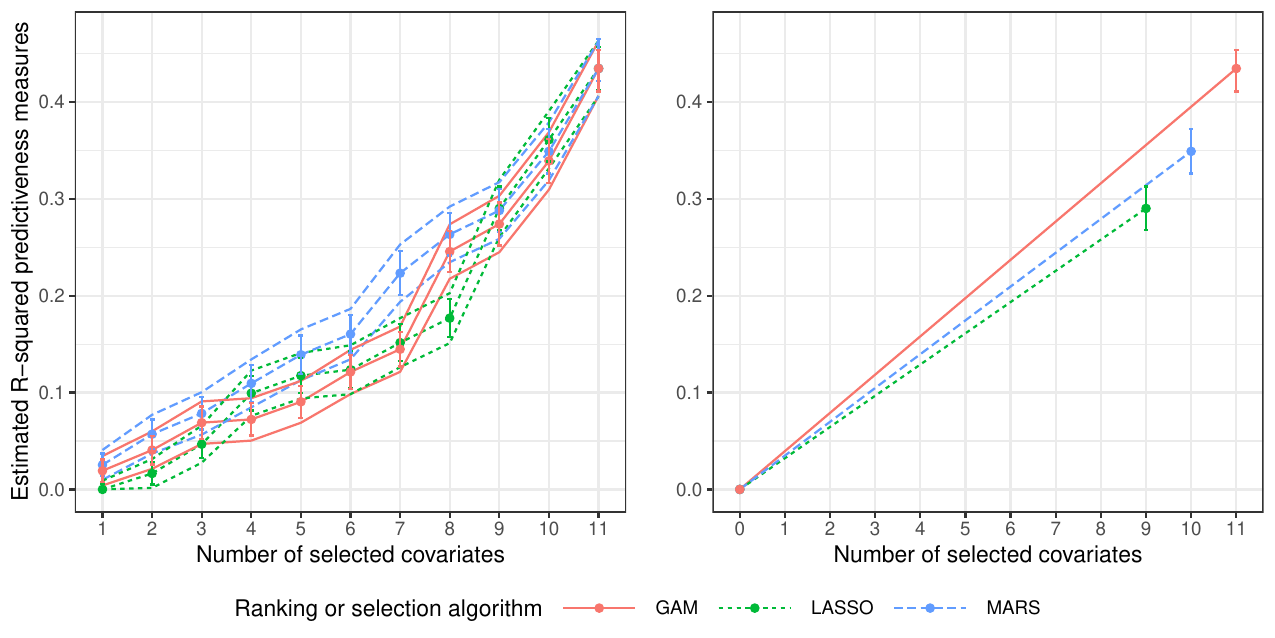}
    \caption{Left panel: Estimated VROC curve using $R$-squared predictiveness metric for prediction of wine rating by physicochemical properties with corresponding 95\% pointwise and uniform confidence intervals. Right panel: Estimated $R$-squared predictiveness measure of selected physicochemical properties for prediction of wine rating with corresponding 95\% pointwise confidence intervals.}
    \label{fig:real data}
\end{figure}

The left panel of \Cref{fig:real data} displays the estimated VROC curves for the $R$-squared predictiveness metric with corresponding 95\% pointwise and confidence sets. All three VROC curves ended at the point $(11, 0.44)$, implying that approximately 44\%  of the variance in wine ratings is accounted for by the 11 physicochemical predictors. None of the three VROC curves were concave, suggesting that the 11 variables contribute roughly equally to the prediction, rather than a single or small set of variables outweighing the importance of the others. In fact, the largest increment for all three curves came when including variables 9, 10, and 11. This suggests that while the last few physicochemical properties alone may have had limited ability to explain the variation in wine quality, they were able to account for a substantial amount of variation when combined with the other physicochemical properties. This indicates the existence of interactions between the physicochemical properties. Furthermore, the MARS ranking algorithm generally had the largest estimated $R$-squared predictiveness of the three algorithms, which suggests that MARS was best able to assess the relative importance of the variables in the presence of interactions. Pairwise tests rejected the null hypothesis that there is no difference in the VROC curves with $p < 10^{-8}$ for all three pairs. We estimate that the AUVROC for MARS was 1.90, (1.75--2.06), the AUVROC for GAM  was 1.62 (1.48--1.76), and the AUVROC for LASSO was 1.60 (1.46--1.74). Pairwise tests rejected the null hypothesis that there is no difference between the AUVROC of MARS and GAM ($p=2.3\times 10^{-7}$) and between MARS and LASSO ($p=1.0\times 10^{-6}$), but not that there was no difference between GAM and LASSO ($p=0.71$).

The right panel of \Cref{fig:real data} displays the $R$-squared predictiveness metric of the selected variables with corresponding 95\% confidence intervals. The GAM selection algorithm selected all variables and had an estimated PPSV of 0.039 (0.037--0.041). MARS selected ten variables and had an estimated PPSV of 0.035 (0.032--0.037). LASSO selected nine variables and had an estimated PPSV of 0.032 (0.029--0.035).

\section{Conclusion}\label{sec:conclusion}

In this paper, we proposed nonparametric, algorithm-agnostic measures of the quality of variable selection and ranking procedures. We proposed plug-in estimators of our measures, and provided conditions under which our estimators are asymptotically linear. Our theoretical results generalize those of \cite{williamson2022general} because our proposed measures are based on the variable importance framework introduced therein, but with random rather than fixed variable sets. We used our asymptotic results to construct large-sample confidence regions for our proposed measures. We also proposed a computationally efficient partial bootstrap procedure to account for finite-sample variability in the variable selection or ranking procedure not accounted for in the asymptotic results.

There are several natural extensions to our work. First, our results were focused on the low-dimensional regime, but the variable selection and ranking algorithms are frequently applied in high-dimensional settings.  An important area of future research is generalizing our results to high-dimensional covariates. Second, some variable selection and ranking procedures may not be asymptotically stable as defined in Section~\ref{sec:stable ranking algorithms}. In these cases, it is not clear how to even define a parameter of interest, or to achieve valid inference for the parameter. This is another important area of future research. Finally, variable selection and ranking is of interest outside of classical regression analysis, such as causal inference and survival analysis \citep{shortreed2017outcome, fan2010high}, and our methods could in principle be extended to these areas as well.

\bibliographystyle{apalike2}
\bibliography{reference.bib}

 \newpage
 \appendix
 
\begin{center}
\textbf{\huge Supplementary Material}
\end{center}

\section{Details on construction of Wald confidence intervals}

Here we provide details for the construction of Wald confidence regions for $\psi_{0,S}$ and $(\psi_{0, R_{[1]}}, \dotsc, \psi_{0, R_{[p]}})$. 

\Cref{thm:asym linearity} implies that $n^{1/2}(\psi_{n,S_n} - \psi_{0,S}) \leadsto N(0, \sigma_0^2)$, where $\sigma_0^2 :=  P_0\left[\dot{V}_0(f_0) - \dot{V}_0(f_{0, -S})\right]^2$. Hence, if $\sigma_n^2$ is a consistent estimator of $\sigma_0^2$, then the Wald confidence interval for $\psi_{0, S}$ given by $\psi_{n, S_n} \pm z_{1-\alpha/2}n^{-1/2}\sigma_n$ has asymptotic coverage level $1-\alpha$, where $z_{\alpha}$ is the lower $\alpha$ quantile of the standard normal distribution. We can estimate $\sigma_0^2$ using the \emph{influence function-based} estimator. We define
\[\dot{V}_n(f_n) := \dot{U}_{\zeta}(\zeta(f_n, \d{P}_n), \eta_n) \left[ \dot\zeta(f_n) - \d{P}_n \dot\zeta(f_n)\right] + \dot{U}_\eta (\zeta(f_n, \d{P}_n), \eta_n) \phi_n \]
as an estimator of $\dot{V}_0(f_0)$, where $\phi_n$ is an estimator of $\phi_0$. We analogously define $\dot{V}_n(f_{n,-S_n})$ and $\dot{V}_n(f_{n, -R_{n, [j]}})$. The influence function-based variance and estimator is then given by
\begin{align*}
    \sigma_n^2 &:=  \d{P}_n \left\{\left[\dot{V}_n(f_n) - \dot{V}_n(f_{n, -S_n})\right] \right\}^2.
\end{align*}

\Cref{thm:asym linearity} can also be used to construct asymptotically valid confidence intervals for each $\psi_{0, R_{[j]}}$ as well as uniformly valid confidence sets for  $(\psi_{0, R_{[1]}}, \dotsc, \psi_{0, R_{[p]}})$, where $R \in \s{R}_0$. By the delta method, \Cref{thm:asym linearity} implies that
\begin{align*}
    & n^{1/2} \left[\left(\psi_{n, R_{n, [1]}}, \dotsc, \psi_{n, R_{n, [p]}}\right)^{\tran} - \left(\psi_{0, R_{[1]}} - \psi_{0, R_{[p]}}\right)^{\tran} \right] \leadsto N_p \left(0, \Sigma_0\right),
\end{align*}
where  $\Sigma_0$ is a $p \times p$ asymptotic covariance matrix with $ij$th entry
\begin{align*}
    (\Sigma_0)_{ij} := & P_0\left\{\left[\dot{V}_0(f_0) - \dot{V}_0\left(f_{0, -R_{[i]}}\right)\right]\left[\dot{V}_0(f_0) - \dot{V}_0\left(f_{0, -R_{[j]}}\right)\right]\right\} 
\end{align*}
for $1 \leq i,j \leq p$ and $R \in \s{R}_0$. If $\Sigma_n$ is a consistent estimator of $\Sigma_0$, then pointwise and equi-precision uniform confidence regions for $(\psi_{0, R_{0, [1]}}, \dotsc, \psi_{0, R_{0, [p]}})$ are given by 
\begin{align*}
    \left(\psi_{n, R_{n, [1]}}, \dotsc, \psi_{n, R_{n, [p]}}\right) &\pm z_{1-\alpha/2}n^{-1/2}\left(\Sigma_{n,11}^{1/2},  \dotsc, \Sigma_{n,pp}^{1/2}\right) \text{ and} \\
    \left(\psi_{n, R_{n, [1]}}, \dotsc, \psi_{n, R_{n, [p]}}\right) &\pm Z_{1-\alpha, n}n^{-1/2}\left(\Sigma_{n,11}^{1/2},  \dotsc, \Sigma_{n,pp}^{1/2}\right),
\end{align*}
respectively, where $Z_{\alpha, n}$ is the lower $\alpha$ quantile of the distribution of $\max_{1\leq j\leq p}|U_j|$, where $U_j \sim N_p(0, D_n^{-1/2} \Sigma_n D_n^{-1/2})$ for $D_n := \mathrm{diag}(\Sigma_n)$. A confidence interval for the AUVROC can be obtained using the delta method. We can again estimate $\Sigma_0$ using an influence function-based estimator given by
\begin{align*}
    (\Sigma_n)_{ij} &:= \d{P}_n\left\{\left[\dot{V}_n(f_n) - \dot{V}_n(f_{n, R_{n, [i]}})\right]\left[\dot{V}_n(f_n) - \dot{V}_n(f_{n, -R_{n, [j]}})\right]\right\}
\end{align*}
for $1 \leq i,j \leq p$.

Similar methods can be used to construct confidence intervals and regions based on the cross-fitting estimators using \Cref{thm: cv asym linear} by replacing $\psi_{n, S_n}$ and $\psi_{n, R_{n, [j]}}$ above with their cross-fitting counterparts and replacing $\sigma_n$ and $\Sigma_n$ with cross-fitting variance and covariance estimators 
\begin{align*}
    \sigma_{n, \mathrm{cv}}^2 &:= \frac{1}{K}\sum_{k=1}^{K} \d{P}_{n, k} \left[\dot{V}_{n, k}(f_{n, k}) - \dot{V}_{n, k}(f_{n, k, -S_n})\right] ^2, \text{ and} \\
    (\Sigma_{n, \mathrm{cv}})_{ij} &:= \frac{1}{K}\sum_{k=1}^{K} \d{P}_{n, k}\left\{\left[\dot{V}_{n, k}(f_{n, k}) - \dot{V}_{n, k}(f_{n, k, -R_{n, [i]}})\right]\left[\dot{V}_{n, k}(f_{n, k}) - \dot{V}_{n, k}(f_{n, k, -R_{n, [j]}})\right]\right\} 
\end{align*}
for $1 \leq i,j \leq p$, where 
\[ \dot{V}_{n, k}(f_n) := \dot{U}_{\zeta}(\zeta(f_{n,k}, \d{P}_{n,k}), \eta_{n,k}) \left[ \dot\zeta(f_{n,k}) - \d{P}_{n,k} \dot\zeta(f_{n,k})\right] + \dot{U}_\eta (\zeta(f_{n,k}, \d{P}_{n,k}), \eta_{n,k}) \phi_{n,k}\]
and $\phi_{n, k}$ is an estimator of $\phi_0$ based on the training data for fold $k$. As with asymptotic linearity, consistency of the cross-fitting variance and covariance estimators does not require constraints on the complexity of the nuisance estimators.

\clearpage 
\section{Proof of theorems}\label{appendix:product_rule}

A one-dimensional path $\{ P_\varepsilon : \varepsilon \in [0, M)\}$ of probability distributions through $P_0$ is called \emph{differentiable in quadratic mean} (DQM) if 
\[ \int \left[ \frac{(dP_\varepsilon)^{1/2} - (dP_0)^{1/2}}{\varepsilon} - \tfrac{1}{2} \dot\ell_0 (dP_0)^{1/2}  \right]^2 = o(1)\]
for some \emph{score function} $\dot\ell_0$. The score function then necessarily satisfies $P_0 \dot\ell_0 = 0$ and $P_0 \dot\ell_0^2 < \infty$. Here and throughout, $dP$ is shorthand for $dP / d\lambda$ for an appropriate dominating measure $\lambda$, and the shorthand is only used when the dominating measure does not change the value of the expression (which is the case of the display above). 

The Hellinger distance between two probability measures is defined as 
\[
    H(Q,P) := \left\{ \int \left[ (dQ)^{1/2} - (dP)^{1/2} \right]^2 \right\}^{1/2}.
\]
The following simple lemmas provide some properties that are useful for justification of product rule in \Cref{thm:pathwise differentiable}.

\begin{lemma}\label{lemma:dqm_hellinger}
    If $\{P_\varepsilon : \varepsilon \in [0, \eta)\}$ is a DQM path through $P_0$, then 
    \[ 
        H(P_\varepsilon, P_0) = \left\{ \int \left[ (dP_\varepsilon)^{1/2} - (dP_0)^{1/2}\right]^2 \right\}^{1/2} = \boundeddet{\varepsilon}.
    \]
\end{lemma}
\begin{proof}[\bfseries{Proof of \Cref{lemma:dqm_hellinger}}]
    We let $\dot\ell_0$ be the score of the path at 0. By the triangle inequality, we have
    \begin{align*}
        &\left\{ \int \left[ (dP_\varepsilon)^{1/2} - (dP_0)^{1/2}\right]^2 \right\}^{1/2}  \\
        &\qquad= \varepsilon \left\{ \int \left[ \frac{(dP_\varepsilon)^{1/2} - (dP_0)^{1/2}}{\varepsilon} - \frac{1}{2}\dot\ell_0 (dP_0)^{1/2} + \frac{1}{2}\dot\ell_0 (dP_0)^{1/2} \right]^2 \right\}^{1/2}  \\
        &\qquad\leq \varepsilon \left\{ \int \left[ \frac{(dP_\varepsilon)^{1/2} - (dP_0)^{1/2}}{\varepsilon} - \frac{1}{2}\dot\ell_0 (dP_0)^{1/2}  \right]^2 \right\}^{1/2} + \frac{\varepsilon}{2}\left\{ \int \dot\ell_0^2 \, dP_0 \right\}^{1/2}.
    \end{align*}
    The first term is $\fasterthandet{\varepsilon}$ by definition of DQM. The second term is $\boundeddet{\varepsilon}$ since $P_0 \dot\ell_0^2 < \infty$.
\end{proof}

\begin{lemma}\label{lemma:mean_bound}
    For any measurable function $h$ and probability measures $P$ and $Q$,
    \[ 
        \left|(Q-P)h \right| \leq \left\{ \| h\|_{L_2(Q)} +  \| h\|_{L_2(P)}\right\} H(Q, P). 
    \]
\end{lemma}
\begin{proof}[\bfseries{Proof of \Cref{lemma:mean_bound}}]
    We have
    \begin{align*}
        \left|(Q-P)h \right|  &= \left|\int h \, d(Q-P)\right| \\
        &= \left|\int h \, \left[(dQ)^{1/2}-(dP)^{1/2} \right] \left[(dQ)^{1/2}+(dP)^{1/2} \right] \right| \\
        &\leq  \left|\int h(dQ)^{1/2} \, \left[(dQ)^{1/2}-(dP)^{1/2} \right] \right| + \left|\int h(dP)^{1/2} \, \left[(dQ)^{1/2}-(dP)^{1/2} \right] \right| \\
        &\leq \left\{ \int h^2 \, dQ \int \left[(dQ)^{1/2}-(dP)^{1/2} \right]^2 \right\}^{1/2} + \left\{ \int h^2 \, dP \int \left[(dQ)^{1/2}-(dP)^{1/2} \right]^2 \right\}^{1/2} \\
        &= \left\{ \| h\|_{L_2(Q)} + \| h \|_{L_2(P)} \right\} H(Q, P).
    \end{align*}
\end{proof}

\begin{lemma}\label{lemma:fixed_fun}
    For any $h \in L_2(P_0)$ and DQM path $\{P_\varepsilon\}$ through $P_0$ with score $\dot\ell_0$ such that there exists $C < \infty$ with $P_\varepsilon h^2 \leq C$ for all $\varepsilon$ small enough, $\left. \frac{\partial}{\partial\varepsilon} P_\varepsilon h \right|_{\varepsilon = 0} = P_0( h\dot\ell_0)$.
\end{lemma}
\begin{proof}[\bfseries{Proof of \Cref{lemma:fixed_fun}}]
    The claim follows if
    \[ 
        \lim_{\varepsilon \to 0} \left[ \frac{ P_\varepsilon h - P_0 h}{\varepsilon} - P_0 (h \dot\ell_0) \right] = 0.
    \]
    Letting $\lambda_{\varepsilon}$ be a common dominating measure of $P_\varepsilon$ and $P_0$ and $p_\varepsilon = dP_\varepsilon / d\lambda_\varepsilon$, $p_{0,\varepsilon} = dP_0 / d\lambda_\varepsilon$, we have
    \begin{align*}
        \frac{ P_\varepsilon h - P_0 h}{\varepsilon} - P_0 (h \dot\ell_0) &= \int h \left[ \frac{p_\varepsilon - p_{0, \varepsilon}}{\varepsilon} - \dot\ell_0 p_{0,\varepsilon} \right] \, d\lambda_\varepsilon \\
         &= \int h \left( p_\varepsilon^{1/2} + p_{0,\varepsilon}^{1/2}\right) \left[\frac{p_\varepsilon^{1/2} - p_{0,\varepsilon}^{1/2}}{\varepsilon} - \tfrac{1}{2}\dot\ell_0 p_{0,\varepsilon}^{1/2} \right] \, d\lambda_\varepsilon + \frac{1}{2} \int h \dot\ell_0 p_{0,\varepsilon}^{1/2} \left( p_\varepsilon^{1/2} - p_{0,\varepsilon}^{1/2}\right)\, d\lambda_\varepsilon.
    \end{align*}
    For the first term, by the Cauchy-Schwarz inequality, 
    \begin{align*}
        &\left|  \int h \left( p_\varepsilon^{1/2} + p_{0,\varepsilon}^{1/2}\right) \left[\frac{p_\varepsilon^{1/2} - p_{0,\varepsilon}^{1/2}}{\varepsilon} - \tfrac{1}{2}\dot\ell_0 p_{0,\varepsilon}^{1/2} \right] \, d\lambda_\varepsilon \right| \\
        &\qquad\leq \left|\int h  p_\varepsilon^{1/2} \left[\frac{p_\varepsilon^{1/2} - p_{0,\varepsilon}^{1/2}}{\varepsilon} - \tfrac{1}{2}\dot\ell_0 p_{0,\varepsilon}^{1/2} \right] \, d\lambda_\varepsilon \right| +\left| \int h  p_{0,\varepsilon}^{1/2} \left[\frac{p_\varepsilon^{1/2} - p_{0,\varepsilon}^{1/2}}{\varepsilon} - \tfrac{1}{2}\dot\ell_0 p_{0,\varepsilon}^{1/2} \right] \, d\lambda_\varepsilon \right| \\
        &\qquad\leq \left( \left[ P_\varepsilon h^2 \right]^{1/2} + \left[P_0 h^2 \right]^{1/2} \right)  \left\{ \int  \left[\frac{p_\varepsilon^{1/2} - p_{0,\varepsilon}^{1/2}}{\varepsilon} - \tfrac{1}{2}\dot\ell_0 p_{0,\varepsilon}^{1/2} \right]^2 \, d\lambda_\varepsilon\right\}^{1/2},
    \end{align*}
    which goes to zero because the first term is bounded as $\varepsilon \to 0$ by assumption and the second term is $o(1)$ because the path is DQM.
    
    For the second term, we have for any sequence of positive numbers $K_\varepsilon$
    \begin{align*}
        \left| \int h \dot\ell_0 p_0^{1/2} \left( p_\varepsilon^{1/2} - p_0^{1/2}\right)\, d\lambda_\varepsilon\right| & \leq  \left| \int_{|h| \leq K_\varepsilon}h \dot\ell_0 p_0^{1/2} \left( p_\varepsilon^{1/2} - p_0^{1/2}\right)\, d\lambda_\varepsilon \right| +  \left| \int_{|h| > K_\varepsilon}h \dot\ell_0 p_0^{1/2} \left( p_\varepsilon^{1/2} - p_0^{1/2}\right)\, d\lambda_\varepsilon \right| \\
         & \leq K_\varepsilon \left| \int \dot\ell_0 p_0^{1/2} \left( p_\varepsilon^{1/2} - p_0^{1/2}\right)\, d\lambda_\varepsilon \right| +  \left| \int\left[ I(|h| > K_\varepsilon)  \dot\ell_0 p_0^{1/2}\right] \left[h p_\varepsilon^{1/2}\right]\, d\lambda_\varepsilon \right|  \\
        &\qquad +   \left| \int\left[ I(|h| > K_\varepsilon)  \dot\ell_0 p_0^{1/2}\right] \left[h p_0^{1/2}\right]\, d\lambda_\varepsilon \right| \\
        & \leq K_\varepsilon \left[  P_0 \dot\ell_0^2  \int \left( p_\varepsilon^{1/2} - p_0^{1/2}\right)^2\, d\lambda_\varepsilon \right]^{1/2} \\
        & \qquad +  \left[ \int_{|h| > K_\varepsilon}  \dot\ell_0^2 \, dP_0\right]^{1/2} \left(  \left[ P_\varepsilon h^2  \right]^{1/2}  + \left[ P_0 h^2  \right]^{1/2} \right).
     \end{align*}
     Since $\dot\ell_0 \in L_2(P_0)$, $P_\varepsilon h^2  \leq C$ for all $\varepsilon$ small enough, and $P_0 h^2 <\infty$, the second term goes to zero as $\varepsilon \to 0$ for any sequence $K_\varepsilon \to \infty$. For the first term, since the path is DQM, $\left[  P_0 \dot\ell_0^2  \int \left( p_\varepsilon^{1/2} - p_0^{1/2}\right)^2\, d\lambda_\varepsilon \right]^{1/2} = \boundeddet{\varepsilon}$ as $\varepsilon \to 0$ by \Cref{lemma:dqm_hellinger}. Therefore, for any path we can choose $K_\varepsilon$ diverging to $\infty$ slowly enough so that the first term is going to zero. For example, set $K_\varepsilon = \left[ \int \left( p_\varepsilon^{1/2} - p_0^{1/2}\right)^2\, d\lambda_\varepsilon \right]^{-1/4}$.
\end{proof}

\thmpathdiff*
\begin{proof}[\bfseries{Proof of \Cref{thm:pathwise differentiable}}]
    We only provide the proof for pathwise differentiability of $P \mapsto V(f_{P, -S_P}, P)$. The proof for $P \mapsto V(f_{P, -R_{P, [j]}}, P), \ j \in \{1, \dotsc, p\}$ follows the same idea, and will be omitted for brevity.
    
    We first show that the map $P \mapsto \zeta(f_{P, -S_P}, P)$ is pathwise differentiable at $P_0$ relative to a nonparametric model with nonparametric efficient influence function at $P_0$ given by $\dot{\zeta}(f) - P_0 \dot\zeta(f)$ under condition~\ref{cond: local continuity}. Let $\{ P_\varepsilon: \varepsilon \in [0, \eta)\}$ be a DQM path in the model with $P_{\varepsilon = 0} = P_0$ and score function $\dot\ell_0$ at $\varepsilon = 0$. We index by $\varepsilon$ to denote $P_\varepsilon$ for simplicity; for example, we let $S_{\varepsilon} := S_{P_\varepsilon}$. We make the following restrictions on the path $\{ P_\varepsilon: \varepsilon \in [0, \eta)\}$. We require that $\| f_{\varepsilon, -S_0} - f_{0, -S_0} \|_{\s{F}}^2 = o(\varepsilon)$, $\dot\zeta(f_{\varepsilon, -S_{\varepsilon}}) = \dot\zeta(f_{\varepsilon, -S_0})$ for small enough $\varepsilon$, $\| \dot{\zeta}(f_{0, -S_0})\|_{L_2(P_\varepsilon)} = O(1)$, and $\| \dot{\zeta}(f_{\varepsilon, -S_0}) - \dot{\zeta}(f_{0, -S_0})\|_{L_2(P_\varepsilon)} = o(1)$.

    We write:
    \begin{align}\label{eq: pathwise decom}
        \begin{split}
         \left. \frac{\partial}{\partial \varepsilon} \zeta(f_{\varepsilon, -S_{\varepsilon}}, P_\varepsilon ) \right|_{\varepsilon=0} & = \lim_{\varepsilon \to 0} \frac{\int \dot{\zeta}(f_{\varepsilon, -S_{\varepsilon}}) \sd P_\varepsilon - \int \dot{\zeta}(f_{0, -S_{0}}) \sd P_0}{\varepsilon}\\
         & = \lim_{\varepsilon \to 0} \frac{\int \dot{\zeta}(f_{\varepsilon, -S_{\varepsilon}}) \sd P_\varepsilon - \int \dot{\zeta}(f_{\varepsilon, -S_{0}}) \sd P_\varepsilon}{\varepsilon} + \lim_{\varepsilon \to 0} \frac{\int \left[\dot{\zeta}(f_{\varepsilon, -S_0}) - \dot{\zeta}(f_{0, -S_0}) \right]\sd P_0 }{\varepsilon}\\
         & \qquad + \lim_{\varepsilon \to 0}  \int \left[\dot{\zeta}(f_{\varepsilon, -S_0}) - \dot{\zeta}(f_{0, -S_0}) \right]\sd (P_\varepsilon - P_0)/\varepsilon \\
         & \qquad + \lim_{\varepsilon \to 0}  \int \left[\dot{\zeta}(f_{0, -S_0})\right] \sd (P_\varepsilon - P_0)/\varepsilon. 
        \end{split}
    \end{align}
    By the assumption that $\dot\zeta(f_{\varepsilon, -S_{\varepsilon}}) = \dot\zeta(f_{\varepsilon, -S_0})$ for small enough $\varepsilon$, the first term on the right-hand side of \eqref{eq: pathwise decom} equals 0. By the assumption that $\| f_{\varepsilon, -S_0} - f_{0, -S_0} \|_{\s{F}}^2 = o(\varepsilon)$, $\norm{f_{\varepsilon, -S_0} - f_{0, -S_0}}_{\s{F}} < \delta$ for $\varepsilon$ small enough for $\delta$ defined in condition~\ref{cond: local continuity}, so by  condition~\ref{cond: local continuity}, 
    \[
       \abs{\int \left[\dot{\zeta}(f_{\varepsilon, -S_0}) - \dot{\zeta}(f_{0, -S_0}) \right]\sd P_0} \leq C \norm{f_{\varepsilon, -S_0} - f_{0, -S_0}}_{\s{F}}^2   = o(\varepsilon).
    \]
    Therefore, the second term on the right-hand side of \eqref{eq: pathwise decom} equals 0. 
    
    Consider the fourth term on the right-hand side of \eqref{eq: pathwise decom}. Denote $h_{\varepsilon, S_0} := \dot{\zeta}(f_{\varepsilon, -S_0}) - \dot{\zeta}(f_{0, -S_0})$.  By Lemmas~\ref{lemma:dqm_hellinger} and~\ref{lemma:mean_bound}, we have 
    \begin{align*}
        \abs{\int \left[\dot{\zeta}(f_{\varepsilon, -S_0}) - \dot{\zeta}(f_{0, -S_0}) \right] \sd (P_\varepsilon - P_0)} & \leq \{ \| h_{\varepsilon, S_0} \|_{L_2(P_\varepsilon)} + \| h_{\varepsilon, S_0} \|_{L_2(P_0)} \} H(P_\varepsilon, P_0) \\
        &= \{ \| h_{\varepsilon, S_0} \|_{L_2(P_\varepsilon)} + \| h_{\varepsilon, S_0} \|_{L_2(P_0)} \} \boundeddet{\varepsilon},
    \end{align*}
     By the continuity of $f \mapsto \dot\zeta(f)$ at $f = f_{0,-S_0}$ assumed in condition~\ref{cond: local continuity} and  and the assumption that $\| f_{\varepsilon, -S_0} - f_{0, -S_0} \|_{\s{F}}^2 = o(\varepsilon)$,  $\|h_{\varepsilon, S_0}\|_{L_2(P_0)} = o(1)$. Also, $\|h_{\varepsilon, S_0}\|_{L_2(P_\varepsilon)} = o(1)$ by assumption. Thus,  which implies that the fourth term on the right-hand side of~\eqref{eq: pathwise decom} equals 0. 
    
    For the last term on the right-hand side of \eqref{eq: pathwise decom}, \Cref{lemma:fixed_fun} and the assumption that $\| \dot{\zeta}(f_{0, -S_0})\|_{L_2(P_\varepsilon)} = O(1)$ implies that 
    \[
        \lim_{\varepsilon \to 0}  \int \dot{\zeta}(f_{0, -S_0}) \sd (P_\varepsilon - P_0)/\varepsilon  = \int \dot{\zeta}(f_{0, -S_0})  \dot\ell_0(x,y) \, dP_0(x,y).
    \]
    Thus, we find that
    \[
        \left. \frac{\partial}{\partial \varepsilon} \zeta\left(f_{\varepsilon, -S_{\varepsilon}}, P_\varepsilon \right) \right|_{\varepsilon=0} = \int \dot{\zeta}\left(f_{0, -S_0}\right) \dot{\ell}(x, y) \sd P_0(x, y),
    \]
     which implies that $P \mapsto \zeta(f_{P, -S_P}, P)$ is pathwise differentiable relative to a nonparametric model with nonparametric efficient influence function given by $\dot{\zeta}\left(f_{0, -S_0}\right)  - P_0 \dot\zeta\left(f_{0, -S_0}\right)$. 

     By conditions~\ref{cond: nuisance}--\ref{cond: derivative} and the chain rule, the map $P \mapsto V(f_{P, S_P}, P)$ is pathwise differentiable relative to a nonparametric model with nonparametric efficient influence function given by 
     \[
        \dot{U}_{\zeta}(\zeta(f_{0, -S}, P_0), \eta_0) \left[ \dot{\zeta}(f_{0, -S_0}) - P_0 \dot{\zeta}(f_{0, -S_0})\right] + \dot{U}_{\eta}(\zeta(f_{0, -S}, P_0), \eta_0) \phi_0.
     \]
\end{proof}

\thmasymlinear*
\begin{proof}[\bfseries{Proof of \Cref{thm:asym linearity}}]
    We only provide the proof of asymptotic linearity of $V_{n, -S_n}$ because the proof of asymptotic linearity of $V_{n, -R_{n, [j]}}$ follows the exact same logic.

    We first show that condition~\ref{cond: convergence rate} and asymptotic stability of $S_n$ imply that $\norm{f_{n, -S_n} - f_{0, -S}}_{\s{F}} = o_{\prob_0}(n^{-1/4})$. For any $\varepsilon >0$ and $S \in \s{S}_0$, we have 
    \begin{align*}
        P_0 \left(n^{1/4}\norm{f_{n, -S_n} - f_{0, -S}}_{\s{F}} \geq \varepsilon \right) \leq P_0 \left(n^{1/4}\norm{f_{n, -S_n} - f_{0, -S}}_{\s{F}} \geq \varepsilon, S_n \in \s{S}_0 \right) + P_0 \left(S_n \notin \s{S}_0 \right).
    \end{align*} 
    The second term goes to zero by asymptotic stability. For the first term, $S_n \in \s{S}_0$ implies that $f_{n, -S_n}= f_{n,-S}$ by the definition of asymptotic stability since $S \in \s{S}_0$. Hence, 
    \[ P_0 \left(n^{1/4}\norm{f_{n, -S_n} - f_{0, -S}}_{\s{F}} \geq \varepsilon, S_n \in \s{S}_0 \right) \leq P_0 \left(n^{1/4}\norm{f_{n, -S} - f_{0, -S}}_{\s{F}} \geq \varepsilon \right),\]
    which goes to zero by~\ref{cond: convergence rate}. 

    We now show that if conditions~\ref{cond: local continuity} and \ref{cond: convergence rate}--\ref{cond: limited complexity} hold for all $S \in \s{S}_0$, then
    \[
        \zeta(f_{n, -S_n}, \d{P}_n)  = \zeta(f_{0, -S}, P_0) + \d{P}_n \left[ \dot{\zeta}(f_{0, -S}) - P_0 \dot{\zeta}(f_{0, -S})\right] + o_{\prob_0}(n^{-1/2}).
    \]
    We consider the decomposition
    \begin{align}
        \begin{split}\label{eq: decomposition}
            &\zeta\left(f_{n, -S_n}, \d{P}_n\right) - \zeta\left(f_{0, -S}, P_0 \right) - \d{P}_n \left[ \dot{\zeta}(f_{0, -S}) - P_0 \dot{\zeta}(f_{0, -S})\right] \\
            &\qquad =\left[\zeta(f_{n, -S_n}, \d{P}_n) - \zeta(f_{n, -S_n}, P_0) - (\d{P}_n - P_0) \dot{\zeta}(f_{n, -S_n})\right] \\
            &\qquad\qquad + \left[\zeta\left(f_{n, -S_n}, P_0\right) - \zeta\left(f_{0, -S}, P_0\right) \right] + (\d{P}_n - P_0) \left[\dot{\zeta}\left(f_{n,-S_n}\right) - \dot{\zeta}\left(f_{0,-S}\right) \right]
        \end{split}
    \end{align}
    We will show that each term in square braces on the right-hand side of \eqref{eq: decomposition} is $o_{\prob_0}(n^{-1/2})$. 

    For the first term on the right-hand side of \eqref{eq: decomposition}, we have 
    \begin{align*}
        &\zeta(f_{n, -S_n}, \d{P}_n) - \zeta(f_{n, -S_n}, P_0) - (\d{P}_n - P_0) \dot{\zeta}(f_{n, -S_n}) \\
        & \qquad = \int \dot{\zeta}(f_{n, -S_n}) \sd \d{P}_n - \int \dot{\zeta}(f_{n, -S_n}) \sd P_0 - (\d{P}_n - P_0) \dot{\zeta}(f_{n, -S_n})\\
        &\qquad =  0.
    \end{align*}
    Consider the second term on the right-hand side of \eqref{eq: decomposition}. For $\delta$ and $C$ defined in condition~\ref{cond: local continuity} and any $\varepsilon >0$, by condition~\ref{cond: local continuity} we have 
    \begin{align*}
        & P_0\left(n^{1/2}\abs{\zeta(f_{n, -S_n}, P_0) - \zeta(f_{0, -S}, P_0)} \geq \varepsilon \right)\\
        & \qquad \leq P_0\left(n^{1/2}\left|\zeta(f_{n, -S_n}, P_0) - \zeta(f_{0, -S}, P_0) \right| \geq \varepsilon, \ \norm{f_{n, -S_n} - f_{0, -S}}_{\s{F}} < \delta \right)\\
        & \qquad\qquad + P_0\left(n^{1/2}\left|\zeta(f_{n, -S_n}, P_0) - \zeta(f_{0, -S}, P_0) \right| \geq \varepsilon, \ \norm{f_{n, -S_n} - f_{0, -S}}_{\s{F}} \geq \delta \right)\\
        & \qquad \leq P_0\left(n^{1/2} \abs{ \int \left[\dot{\zeta}(f_{n, -S_n}) - \dot{\zeta}(f_{0, -S})\right] \sd P_0} \geq \varepsilon, \ \norm{f_{n, -S_n} - f_{0, -S}}_{\s{F}} < \delta \right)\\
        & \qquad \qquad + P_0\left(\norm{f_{n, -S_n} - f_{0, -S}}_{\s{F}} \geq \delta \right)\\
        & \qquad \leq P_0\left(n^{1/2}C \norm{f_{n, -S_n} - f_{0, -S}}_{\s{F}}^2   \geq \varepsilon \right) + P_0\left(\norm{f_{n, -S_n} - f_{0, -S}}_{\s{F}} \geq \delta \right),
    \end{align*}
    which converges to 0 because $\norm{f_{n, -S_n} - f_{0, -S}}_{\s{F}} = o_{\prob_0}(n^{-1/4})$ as shown above. Therefore, the second term on the right-hand side is $o_{\prob_0}(n^{-1/2})$.

    Consider the third term on the right-hand side of \eqref{eq: decomposition}. We have 
    \begin{align*}
        n^{1/2}(\d{P}_n - P_0) \left[\dot{\zeta}(f_{n,-S_n}) - \dot{\zeta}(f_{0,-S})\right]  =  \d{G}_n \left[\dot{\zeta}(f_{n,-S_n}) - \dot{\zeta}(f_{0,-S})\right].
    \end{align*}
    By condition~\ref{cond: local continuity}, $\dot{\zeta}$ is continuous at $f_{0,-S}$, and as shown above, $\|f_{n,-S_n} - f_{0, -S}\|_{\s{F}} = o_{\prob_0}(1)$. This implies that $\|\dot{\zeta}(f_{n,-S_n}) - \dot{\zeta}(f_{0,-S})\|_{L_2(P_0)} = o_{\prob_0}(1)$. Hence, by condition~\ref{cond: limited complexity} and Lemma~19.24 of \cite{van2000asymptotic}, $\d{G}_n \left[\dot{\zeta}(f_{n,-S_n}) - \dot{\zeta}(f_{0,-S})\right] = o_{\prob_0}(n^{-1/2})$.  Therefore, the third term on the right-hand side is $o_{\prob_0}(n^{-1/2})$. 
    
    We have now shown that $\zeta(f_{n, -S_n}, \d{P}_n)  = \zeta(f_{0, -S}, P_0) + (\d{P}_n - P_0) \dot{\zeta}(f_{0, -S}) + o_{\prob_0}(n^{-1/2})$. By conditions~\ref{cond: nuisance}--\ref{cond: derivative} and the delta method, we therefore conclude that 
    \[
        V_{n, -S_n}  = V(f_{0, -S}, P_0) + \d{P}_n \dot{V}_0(f_{0, -S}) + o_{\prob_0}(n^{-1/2}),
    \]
    where $\dot{V}_0(f_{0, -S}) = \dot{U}_{\zeta}(\zeta(f_{0, -S}, P_0), \eta_0) \left[ \dot{\zeta}(f_{0, -S_0}) - P_0\dot{\zeta}(f_{0, -S_0})\right] + \dot{U}_{\eta}(\zeta(f_{0, -S}, P_0), \eta_0)\phi_0$.
\end{proof}

We next introduce lemmas that we use to establish Theorem~\ref{thm: cv asym linear}. 

\begin{lemma} \label{lemma: cv tool 2}
    Suppose that $Z_1, \dotsc, Z_n$ is a sequence of i.i.d. random elements defined on a measurable space $\s{X}\times \s{Y}$ with Borel $\sigma$-field $\s{B}$. Let $h$ be a real-valued function defined on the measurable space $(\s{X}\times \s{Y}, \s{B})$. Then $\E_0 \abs{(\d{P}_n - P_0) h} \leq c n^{-1/2} \| h \|_{L_2(P_0)}$, where $c := (24\log 2)^{1/2}$.
\end{lemma}
\begin{proof}[\bfseries{Proof of Lemma~\ref{lemma: cv tool 2}}]
    The proof can be found in the proof of Lemma~1 of \cite{williamson2022general}.
\end{proof}

\thmcvclassical*
\begin{proof}[\bfseries{Proof of Theorem~\ref{thm: cv asym linear}}]
    As above, we only provide the proof for asymptotic linearity of $\frac{1}{K} \sum_{k=1}^{K} V_{n, k, -S_n}$ because the proof for $\frac{1}{K} \sum_{k=1}^{K} V_{n, k, -R_{n, [j]}}, \ j \in \{1, \dotsc, p\}$ follows the same logic.
   
    By the same argument used in the proof of Theorem~\ref{thm:asym linearity}, condition~\ref{cond: cv convergence rate} and asymptotic stability of $S_n$ imply that $\max_k \norm{f_{n, k, -S_n} - f_{0, -S}}_{\s{F}} = o_{\prob_0}(n^{-1/4})$. We first show that 
    \[
        \zeta(f_{n, k, -S_n}, \d{P}_{n, k})  = \zeta(f_{0, -S}, P_0) + (\d{P}_{n, k} - P_0) \dot{\zeta}(f_{0, -S}) + o_{\prob_0}(n_k^{-1/2}),
    \] 
    for each $k \in \{1, \dotsc, K\}$. We consider the decomposition
    \begin{align}\label{eq: cv decomposition}
        \begin{split}
            & \zeta(f_{n, k, -S_n}, \d{P}_{n, k}) - \zeta(f_{0, -S}, P_0) - (\d{P}_{n, k} - P_0) \dot{\zeta}(f_{0, -S})   \\
            &\qquad = \left[\zeta(f_{n, k, -S_n}, \d{P}_{n, k}) - \zeta(f_{n, k, -S_n}, P_0) - (\d{P}_{n, k} - P_0) \dot{\zeta}(f_{n, k, -S_n})\right] \\
            & \qquad \qquad + \left[\zeta(f_{n, k, -S_n}, P_0) - \zeta(f_{0, -S}, P_0)\right] + (\d{P}_{n, k} - P_0) \left[\dot{\zeta}(f_{n,k,-S_n}) - \dot{\zeta}(f_{0,-S}) \right].
        \end{split}
    \end{align}
    We now show that each term in square braces on the right-hand side of \eqref{eq: cv decomposition} is $o_{\prob_0}(n_k^{-1/2})$.
    
    Consider the first term on the right-hand side of \eqref{eq: cv decomposition}. We have 
    \begin{align*}
        & \zeta(f_{n, k, -S_n}, \d{P}_{n, k}) - \zeta(f_{n, k, -S_n}, P_0) - (\d{P}_{n, k} - P_0) \dot{\zeta}(f_{n, k, -S_n})\\
        & \qquad = \int \dot{\zeta}(f_{n, k, -S_n}) \sd \d{P}_{n, k} - \int \dot{\zeta}(f_{n, k, -S_n}) \sd P_0 - (\d{P}_{n, k} - P_0) \dot{\zeta}(f_{n, k, -S_n})\\
        & \qquad = 0.
    \end{align*}
    Hence the first term on the right-hand side of \eqref{eq: decomposition} is $o_{\prob_0}(n_k^{-1/2})$.

    Consider the second term on the right-hand side of \eqref{eq: cv decomposition}. For $\delta$ and $C$ defined in condition~\ref{cond: local continuity} and any $\varepsilon >0$, by condition~\ref{cond: local continuity} we have 
    \begin{align*}
        & P_0\left(n_k^{1/2} \left|\zeta(f_{n, k, -S_n}, P_0) - \zeta(f_{0, -S}, P_0) \right| \geq \varepsilon \right)\\
        & \qquad \leq P_0\left(n_k^{1/2}\left| \zeta(f_{n, k, -S_n}, P_0) - \zeta(f_{0, -S}, P_0)\right| \geq \varepsilon, \ \norm{f_{n, k, -S_n} - f_{0, -S}}_{\s{F}} < \delta \right)\\
        & \qquad\qquad + P_0\left(n_k^{1/2}\left|\zeta(f_{n, k, -S_n}, P_0) - \zeta(f_{0, -S}, P_0)\right| \geq \varepsilon, \ \norm{f_{n, k, -S_n} - f_{0, -S}}_{\s{F}} \geq \delta \right)\\
        & \qquad \leq P_0\left(n_k^{1/2} \left| \int\left[  \dot{\zeta}(f_{n, k, -S_n}) - \dot{\zeta}(f_{0, -S})\right] \sd P_0\right| \geq \varepsilon, \ \norm{f_{n, k, -S_n} - f_{0, -S}}_{\s{F}} < \delta \right) \\
        & \qquad \qquad + P_0\left(\norm{f_{n, k, -S_n} - f_{0, -S}}_{\s{F}} \geq \delta \right)\\
        & \qquad \leq P_0\left(n_k^{1/2}C \norm{f_{n, k, -S_n} - f_{0, -S}}_{\s{F}}^2  \geq \varepsilon \right) + P_0\left(\norm{f_{n, k, -S_n} - f_{0, -S}}_{\s{F}} \geq \delta \right),
    \end{align*}
    which converge to 0 because $\norm{f_{n, k, -S_n} - f_{0, -S}}_{\s{F}} = o_{\prob_0}(n_k^{-1/4})$ as shown above. Therefore, the second term on the right-hand side is $o_{\prob_0}(n_k^{-1/2})$.

   For the third term on the right-hand side of \eqref{eq: cv decomposition}, we have 
    \begin{align*}
        &P_0\left( n_k^{1/2}\abs{(\d{P}_{n, k} - P_0) \left[\dot{\zeta}(f_{n, k, -S_n}) - \dot{\zeta}(f_{0,-S})\right]} \geq \varepsilon\right) \\
        & \qquad =  P_0\left( \abs{\d{G}_{n, k} \left[\dot{\zeta}(f_{n, k, -S_n}) - \dot{\zeta}(f_{0,-S})\right]} \geq \varepsilon\right)\\
        & \qquad \leq P_0\left( \abs{\d{G}_{n, k} \left[\dot{\zeta}(f_{n, k, -S_n}) - \dot{\zeta}(f_{0,-S})\right]} \geq \varepsilon, \ S_n \in \s{S}_0 \right) + P_0\left( S_n \not \in \s{S}_0 \right)\\
        & \qquad = P_0\left( \abs{\d{G}_{n, k} \left[\dot{\zeta}(f_{n, k, -S}) - \dot{\zeta}(f_{0,-S})\right]} \geq \varepsilon \right) + P_0\left( S_n \not \in \s{S}_0 \right).
    \end{align*}
    The second term on the right-hand side of previous display converges to 0 by assumption of asymptotic stability. Consider the first term on the right-hand side of previous display. Denote the training set for fold $k$ as $\s{D}_{n, k}$. By Chebyshev's inequality and the tower property, for each $S \in \s{S}_0$ we have 
    \begin{align*}
        P_0 \left( \abs{\d{G}_{n, k} [\dot{\zeta}(f_{n, k, -S}) - \dot{\zeta}(f_{0, -S})]} \geq \varepsilon \right) &\leq \frac{E_0 \abs{\d{G}_{n, k} [\dot{\zeta}(f_{n, k, -S}) - \dot{\zeta}(f_{0, -S})]} }{\varepsilon} \\
        &=  \frac{E_0 \left\{ E_0 \left[\abs{\d{G}_{n, k} \left[\dot{\zeta}(f_{n, k, -S}) - \dot{\zeta}(f_{0, -S})\right]} \mid \s{D}_{n, k} \right]\right\} }{\varepsilon}
    \end{align*}
    Conditional on the training data, $\dot{\zeta}(f_{n, k, -S}) - \dot{\zeta}(f_{0, -S})$ is a fixed function, so by Lemma~\ref{lemma: cv tool 2},
    \begin{align*}
       \frac{E_0 \left\{ E_0 \left[\abs{\d{G}_{n, k} \left[\dot{\zeta}(f_{n, k, -S}) - \dot{\zeta}(f_{0, -S})\right]} \mid \s{D}_{n, k} \right]\right\} }{\varepsilon}  &\leq \frac{c_2 E_0\left[ \| \dot{\zeta}(f_{n, k, -S}) - \dot{\zeta}(f_{0, -S}) \|_{L_2(P_0)} \right]}{\varepsilon},
    \end{align*}
    which is $o(1)$ for each $\varepsilon > 0$ by condition~\ref{cond: cv weak consistency}. Therefore, the third term on the right-hand side of \eqref{eq: cv decomposition} is $o_{\prob_0}(n_k^{-1/2})$.
    
    Therefore, $\zeta(f_{n, k, -S_n}, \d{P}_{n, k})  = \zeta(f_{0, -S}, P_0) + (\d{P}_{n, k} - P_0) \dot{\zeta}(f_{0, -S}) + o_{\prob_0}(n_k^{-1/2})$ for each $k \in \{1, \cdots, K\}$. By conditions~\ref{cond: nuisance}--\ref{cond: derivative} and delta method, we therefore have $V_{n, k, -S_n} - V(f_{0, -S}, P_0) =  \d{P}_{n,k} \dot{V}_0(f_{0, -S}) +  o_{\prob}({n_k}^{-1/2})$. We can now write
    \[
        \frac{1}{K} \sum_{k=1}^{K} V_{n, k, -S_n} - V(f_{0, -S}, P_0) = \frac{1}{K} \sum_{k=1}^{K} \d{P}_{n,k} \dot{V}_0(f_{0, -S}) + \frac{1}{K} \sum_{k=1}^{K} o_{\prob}({n_k}^{-1/2}).
    \]
    Since $\max_k \abs{\frac{n}{K n_k} - 1} = o(1)$ by assumption, we can choose $n$ large enough that $\max_k \abs{\frac{n}{K n_k} - 1} \leq 1$, which implies that $\min_k n_k \geq n / (2K)$. Hence, $\max_k n / n_k = O(1)$, so $o_{\prob}({n_k}^{-1/2}) = o_{\prob}({n}^{-1/2})$, which implies that $\frac{1}{K} \sum_{k=1}^{K} o_{\prob}({n_k}^{-1/2}) =  o_{\prob}({n}^{-1/2})$. Note also that $ \abs{\frac{1}{K} \sum_{k=1}^{K} \d{P}_{n,k} \dot{V}_0(f_{0, -S}) - \d{P}_n \dot{V}_0(f_{0, -S}) } \leq \max_k \abs{\frac{n}{K n_k} - 1} \d{P}_n \dot{V}_0(f_{0, -S}) = o(1) O_{\prob_0}(n^{-1/2}) = o_{\prob_0}(n^{-1/2})$.

    We therefore have 
    \[
        \frac{1}{K} \sum_{k=1}^{K} V(f_{n, k, -S_n}, \d{P}_{n,k}) = V(f_{0, -S}, P_0) + \d{P}_{n} \dot{V}_0(f_{0, -S}) + o_{\prob_0}(n^{-1/2}).
    \]
\end{proof}

\clearpage

\section{Verification of condition~\ref{cond: local continuity}}

Here, we present proofs verifying that the conditions stated in Section~\ref{sec:efficiency} imply that condition~\ref{cond: local continuity} holds for the three example predictiveness metrics.

\subsection{R-squared}
    For simplicity, we denote $f_{0, -S}$ and $f_{-S}$ as $f_0$ and $f$, where $f_{-S}$ is an arbitrary function in $\s{F}_{-S}$. Since $\zeta(f, P) := E_P[Y - f(X)]^2$ is linear in $P$, we have $\dot{\zeta}(f)(x, y) = [y - f(x)]^2$. We next show that $f \mapsto \zeta(f)$ is continuous at $f_0$. By H\"{o}lder's inequality, we have 
    \begin{align*}
        \norm{\dot{\zeta}(f) - \dot{\zeta}(f_0)}_{L_2(P_0)} & =  \norm{[y-f]^2 - [y-f_0]^2}_{L_2(P_0)} \\
        & =  \norm{[f_0 - f][2y - f_0 - f]}_{L_2(P_0)} \\
        &\leq  \norm{f_0 - f}_{L_{p}(P_0)}\norm{2y - f_0 - f}_{L_{q}(P_0)}\\
        &\leq \norm{f_0 - f}_{L_{p}(P_0)}\left[ \norm{y - f_0}_{L_{q}(P_0)} + \norm{y - f}_{L_{q}(P_0)} \right],
    \end{align*}
    where $q \in (2, \infty]$ is such that $\int |y - f(x)|^q \sd P_0(y,x) < \infty$ for all $f$ in a neighborhood of $f_0$, and $p \in [2, \infty)$ is such that $2/p + 2/q = 1$.  Hence, the first statement in condition~\ref{cond: local continuity} holds with $\|\cdot\|_{\s{F}} = \|\cdot \|_{L_{p}(P_0)}$. As a special case where $p = 2$ and $q = \infty$, the map $f \mapsto \dot{\zeta}(f)$ is continuous with respect to $L_2(P_0)$-norm provided $|Y - f(X)|$ is bounded $P_0$-almost surely for all $f$ in a neighborhood of $f_0$.

    We next verify that the second statement in condition~\ref{cond: local continuity} holds. We recall that $f_{0, -S}(x) = E_0(Y \mid X_{-S} = x_{-S})$.  By the triangle inequality and law of total expectation,  
    \begin{align*}
        \abs{P_0 \left[ \dot{\zeta}(f) - \dot{\zeta}(f_0) \right]} &= \abs{E_0 \left[ (f_0 - f) (2Y - 2f_0)\right] + E_0 ( f - f_0)^2} \\
        &\leq \abs{ E_0 \left\{ E_0 \left[ (f_0 - f) (2Y - 2f_0) \mid X_{-S} \right]  \right\} } + E_0 ( f - f_0)^2\\
        & = \abs{ E_0 \left\{  (f_0 - f) (2E_0[Y \mid X_{-S} ]  - 2f_0) \right\} } + E_0 ( f - f_0)^2\\
        & = \norm{f - f_0}_{L_2(P_0)}^2\\
        & \leq \norm{f - f_0}_{L_{p}(P_0)}^2, 
    \end{align*}
    for any $p \in [2, \infty]$. This implies the second statement in condition~\ref{cond: local continuity}.

\subsection{Deviance}
    For simplicity, we denote $f_{0, -S}$ and $f_{-S}$ as $f_0$ and $f$, where $f_{-S}$ is an arbitrary function falls in $\s{F}_{-S}$. Since $\zeta(f, P) := E_P[\nu(Y, f(X))]$ is linear in $P$, we have $\dot{\zeta}(f)(x, y) = \nu(y, f(x))$.  We next show that $f \mapsto \dot\zeta(f)$ is continuous at $f_0$. We have
    \begin{align*}
        \dot{\zeta}(f) - \dot{\zeta}(f_0) &= \nu(y, f(x)) - \nu(y, f_0(x))  \\
        & = y \log\frac{f(x)}{f_0(x)} + (1-y)\log \frac{1-f(x)}{1-f_0(x)}.
    \end{align*}
    We now consider the Taylor expansion of the function $u \mapsto g_{x, y}(u) := y \log\frac{u}{f_0(x)} + (1-y)\log \frac{1-u}{1-f_0(x)}$ at $u = f_0(x)$. By Taylor's Theorem with the Lagrange form of the remainder, we have
    \begin{align*}
        g_{x, y}(f(x)) &= g_{x, y}(f_0(x)) + [f(x) - f_0(x)]  \left.\frac{\sd g_{x, y}(u)}{\sd u} \right|_{u = \eta_{x, y}} =  [f(x) - f_0(x)] \left[ \frac{y}{\eta_{x, y}} - \frac{1-y}{1-\eta_{x, y}}\right]
    \end{align*}
    where $\eta_{x, y}$ lies between $f(x)$ and $f_0(x)$. We therefore have 
    \begin{align*}
        \norm{\dot{\zeta}(f) - \dot{\zeta}(f_0)}_{L_2(P_0)} = &\norm{[f - f_0] \left[ \frac{y}{\eta_{x, y}} - \frac{1-y}{1-\eta_{x, y}}\right]}_{L_2(P_0)}\\
        \leq & \sup_{x, y} \left| \frac{y}{\eta_{x, y}} - \frac{1-y}{1-\eta_{x, y}}\right| \norm{f - f_0}_{L_2(P_0)}.
    \end{align*}
    By assumption, there exists $\alpha \in (0, 1/2)$ such that $f_0(X), f(X) \in (\alpha, 1-\alpha)$ almost surely, which implies that $\eta_{X, Y} \in (\alpha, 1-\alpha)$ almost surely as well. Hence, 
    \[
        \sup_{x, y} \left| \frac{y}{\eta_{x, y}} - \frac{1-y}{1-\eta_{x, y}}\right| \leq \sup_{x, y} \left| \frac{y}{\eta_{x, y}(1-\eta_{x, y})}\right| + \sup_{x, y}\left| \frac{\eta_{x, y}}{\eta_{x, y}(1-\eta_{x, y})}\right| \leq \frac{1}{\alpha(1-\alpha)} + \frac{1}{\alpha} = \frac{2-\alpha}{\alpha(1-\alpha)},
    \]
    which implies that $f \mapsto \dot\zeta(f)$ is continuous at $f_0$.
    
    We next verify that the second statement in condition~\ref{cond: local continuity} holds. We recall that $f_{0, -S}(x) = E_0(Y \mid X_{-S} = x_{-S})$.  By definition of $\dot{\zeta}$, $\nu$ and the law of total expectation,  
    \begin{align*}
        \abs{P_0 \left[ \dot{\zeta}(f) - \dot{\zeta}(f_0) \right]} &= \abs{ E_0 \left\{Y \log\frac{f(X)}{f_0(X)} + (1-Y)\log \frac{1-f(X)}{1-f_0(X)} \right\} }\\
        & = \abs{ E_0 \left\{f_0(X) \log\frac{f(X)}{f_0(X)} + [1-f_0(X)]\log \frac{1-f(X)}{1-f_0(X)} \right\} }.
    \end{align*}
    We now consider the Taylor expansion of the function $u \mapsto h_x(u) := f_0(x) \log\frac{u}{f_0(x)} + [1-f_0(x)]\log \frac{1-u}{1-f_0(x)}$ at $u = f_0(x)$. By Taylor's Theorem with the Lagrange form of the remainder, we have
    \begin{align*}
        h_x(f(x)) &= h_x(f_0(x)) + [f(x) - f_0(x)]  \left.\frac{\sd h_x(u)}{\sd u} \right|_{u = f_0(x)} + \frac{1}{2} [f(x) - f_0(x)]^2 \left.\frac{\sd^2 h_x(u)}{\sd u^2} \right|_{u = \xi_x}\\
        & = 0 + [f(x) - f_0(x)] \times 0 - \frac{1}{2} [f(x) - f_0(x)]^2 \left[ \frac{f_0(x)}{\xi_x^2} + \frac{1-f_0(x)}{(1-\xi_x)^2}\right]\\
        & = - \frac{1}{2} [f(x) - f_0(x)]^2 \left[ \frac{f_0(x) - \xi_x}{\xi_x(1-\xi_x)}\right]^2,
    \end{align*}
    where $\xi_x$ lies between $f(x)$ and $f_0(x)$. We therefore have 
    \begin{align*}
        \abs{P_0 \left[ \dot{\zeta}(f) - \dot{\zeta}(f_0) \right]} &= \abs{ E_0 \left\{ - \frac{1}{2} [f(x) - f_0(x)]^2 \left[ \frac{f_0(x) - \xi_x}{\xi_x(1-\xi_x)}\right]^2 \right\} }\\
        & \leq \frac{1}{2} \sup_x \left| \frac{f_0(x) - \xi_x}{\xi_x(1-\xi_x)}\right|^2 \norm{f - f_0}_{L_2(P_0)}^2.
    \end{align*}
    By assumption, there exists $\alpha \in (0, 1/2)$ such that $f_0(X), f(X) \in (\alpha, 1-\alpha)$ almost surely, which implies that $\xi_X \in (\alpha, 1-\alpha)$ almost surely as well. Hence,
    \[
        \sup_x \left| \frac{f_0 - \xi}{\xi_x(1-\xi_x)} \right| \leq \sup_x \left| \frac{f_0(x)}{\xi_x(1-\xi_x)} \right| + \sup_x \left| \frac{\xi_x}{\xi_x(1-\xi_x)}\right|  \leq 2 \frac{1-\alpha}{\alpha(1-\alpha)} = \frac{2}{\alpha},
    \]
    which shows the second statement in condition~\ref{cond: local continuity}.

\subsection{Classification accuracy}
     For simplicity, we denote $f_{0, -S}$ and $f_{-S}$ as $f_0$ and $f$, where $f_{-S}$ is an arbitrary function falls in $\s{F}_{-S}$. Since $\zeta(f, P) := E_PI\{Y = f(X)\}$ is linear in $P$, we have $\dot{\zeta}(f)(x, y) = I\{y = f(x)\}$. We next show that $f \mapsto \zeta(f)$ is continuous at $f_0$. We note that 
    \begin{align*}
        &\norm{\dot{\zeta}(f) - \dot{\zeta}(f_0)}_{L_2(P_0)}^2 \\
        & \qquad = \norm{I\{y = f(x)\} - I\{y = f_0(x)\}}_{L_2(P_0)}^2 \\
        & \qquad = \int I\{y = f(x)\} + I\{y = f_0(x)\} - 2 I\{y = f(x), y=f_0(x)\}  \sd P_0\\
        & \qquad = P_0(Y=f(X), Y=f_0(X)) + P_0(Y=f(X), Y\neq f_0(X))\\
        & \qquad \qquad + P_0(Y=f_0(X), Y=f(X)) + P_0(Y=f_0(X), Y\neq f(X))\\
        & \qquad \qquad - 2P_0(Y=f(X), Y=f_0(X))\\
        & \qquad = P_0(Y=f(X), Y\neq f_0(X)) + P_0(Y=f_0(X), Y\neq f(X))\\
        & \qquad = P_0(Y=1, f(X)=1, f_0(X)=0) + P_0(Y=0, f(X)=0, f_0(X)=1)\\
        & \qquad \qquad + P_0(Y=0, f(X)=1, f_0(X)=0) + P_0(Y=1, f(X)=0, f_0(X)=1)\\
        & \qquad = \left\{ P_0\left(Y=1 \mid f(X)=1, f_0(X)=0\right) + P_0\left(Y=0 \mid f(X)=1, f_0(X)=0\right) \right\} p_1\\
        & \qquad \qquad + \left\{ P_0\left(Y=1 \mid f(X)=0, f_0(X)=1\right) + P_0\left(Y=0 \mid f(X)=0, f_0(X)=1\right) \right\} p_2\\
        & \qquad = p_1 + p_2,
    \end{align*}
    where $p_1 := P_0(f(X)=1, f_0(X)=0)$ and $p_2 := P_0(f(X)=0, f_0(X)=1)$. Recall the definition of $f(x) = I\{\mu(x) > 0.5\}$, we therefore have 
    \begin{align*}
        \norm{\dot{\zeta}(f) - \dot{\zeta}(f_0)}_{L_2(P_0)}^2 = & P_0(\mu(X) > 1/2 \geq \mu_0(X)) + P_0(\mu_0(X) > 1/2 \geq \mu(X))\\
        \leq & 2 P_0( |\mu_0(X) - 1/2| \leq |\mu(X) - \mu_0(X)| )\\
        \leq & 2 \int \frac{|\mu(X) - \mu_0(X)|}{|\mu_0(X) - 1/2|} \sd P_0\\
        \leq & 2 \norm{\mu - \mu_0}_\infty \int \frac{1}{|\mu_0(X) - 1/2|} \sd P_0, 
    \end{align*}
    which implies the first statement in condition~\ref{cond: local continuity} provided $P_0|\mu_0 - 1/2|^{-1}  < \infty$.
     
     We next verify that the second statement in condition~\ref{cond: local continuity} holds. Recall that $f_{0, -S} = I\{\mu_{0, -S} > 0.5\}$, where $\mu_{0, -S}(x_{-S}) = E_0(Y \mid X_{-S} = x_{-S})$. We have 
    \begin{align*}
        & \abs{P_0 \left[ \dot{\zeta}(f) - \dot{\zeta}(f_0) \right]} \\
        & \qquad = \abs{P_0(Y = f(X)) - P_0(Y = f_0(X))}\\
        & \qquad = \abs{P_0(Y = f(X), Y \neq f_0(X)) - P_0(Y = f_0(X), Y \neq f(X))}\\
        & \qquad = \left|P_0(Y=1, f(X)=1, f_0(X)=0) + P_0(Y=0, f(X)=0, f_0(X)=1) \right.\\
        & \qquad  \qquad \left. - P_0(Y=0, f(X)=1, f_0(X)=0) - P_0(Y=1, f(X)=0, f_0(X)=1)\right|\\
        & \qquad  = \left|\left[ P_0 \left(Y=1 \, \Big| \, \mu(X) > \tfrac{1}{2} \geq \mu_0(X)\right) - P_0\left(Y=0 \, \Big| \, \mu(X) > \tfrac{1}{2} \geq \mu_0(X)\right) \right] P_0\left(\mu(X) > \tfrac{1}{2} \geq \mu_0(X)\right)  \right.\\
        & \qquad \qquad + \left.  \left[ P_0\left(Y=0 \, \Big| \, \mu_0(X) > \tfrac{1}{2} \geq \mu(X)\right) - P_0 \left(Y=1 \, \Big| \, \mu_0(X) > \tfrac{1}{2} \geq \mu(X)\right) \right] P_0\left(\mu_0(X) > \tfrac{1}{2} \geq \mu(X)\right)\right|\\
        & \qquad  = \left|\left[ 2P_0\left(Y=1 \, \Big| \, \mu(X) > \tfrac{1}{2} \geq \mu_0(X)\right) - 1 \right] P_0\left(\mu(X) > \tfrac{1}{2} \geq \mu_0(X)\right)  \right.\\
        & \qquad \qquad + \left.  \left[ 2P_0\left(Y=0 \, \Big| \, \mu_0(X) > \tfrac{1}{2} \geq \mu(X)\right) - 1 \right] P_0\left(\mu_0(X) > \tfrac{1}{2} \geq \mu(X)\right)\right|.
    \end{align*}
    We also have
    \begin{align*}
        \left| P_0\left(Y=1 \, \Big| \, \mu(X) > \tfrac{1}{2} \geq \mu_0(X)\right) - \tfrac{1}{2}\right| & = \left| E_0 \left\{\mu_0(X) - \tfrac{1}{2} \, \Big| \, \mu(X) > \tfrac{1}{2} \geq \mu_0(X)\right\} \right| \\
        & \leq \norm{\mu - \mu_0 }_{\infty}.
    \end{align*}
    
    Similarly, we have
    \begin{align*}
       \left| P_0\left(Y=0 \mid \mu_0(X) > \tfrac{1}{2} \geq \mu(X)\right) - \tfrac{1}{2} \right| &\leq \norm{\mu - \mu_0 }_{\infty}.
    \end{align*}
    In addition, by Chebyshev's inequality, 
    \begin{align*}
        P_0\left(\mu_0(X) > \tfrac{1}{2} \geq \mu(X)\right) & \leq P_0\left(\abs{\mu_0 - \tfrac{1}{2}} < \abs{\mu - \mu_0}\right) \leq  P_0 \left[\frac{|\mu - \mu_0|}{|\mu_0 - 1/2|} \right]\\
        & \leq \norm{\mu - \mu_0}_\infty P_0\left[ \frac{1}{|\mu_0 - 1/2|}\right] ,
    \end{align*}
    and similarly, $P_0\left(\mu(X) > \frac{1}{2} \geq \mu_0(X)\right) \leq \norm{\mu - \mu_0}_\infty P_0 \frac{1}{|\mu_0 - 1/2|}$. We therefore have 
    \begin{align*}
        \abs{P_0 \left[ \dot{\zeta}(f) - \dot{\zeta}(f_0) \right]} \leq \norm{\mu - \mu_0}_\infty^2 \int \frac{1}{|\mu_0(X) - 1/2|} \sd P_0.
    \end{align*}
    Hence, the second statement in condition~\ref{cond: local continuity} holds provided $P_0|\mu_0 - 1/2|^{-1}  < \infty$.

\newpage
\section{Additional results from the numerical studies}

\begin{figure}[!htbp]
    \centering
    \includegraphics[width=\linewidth]{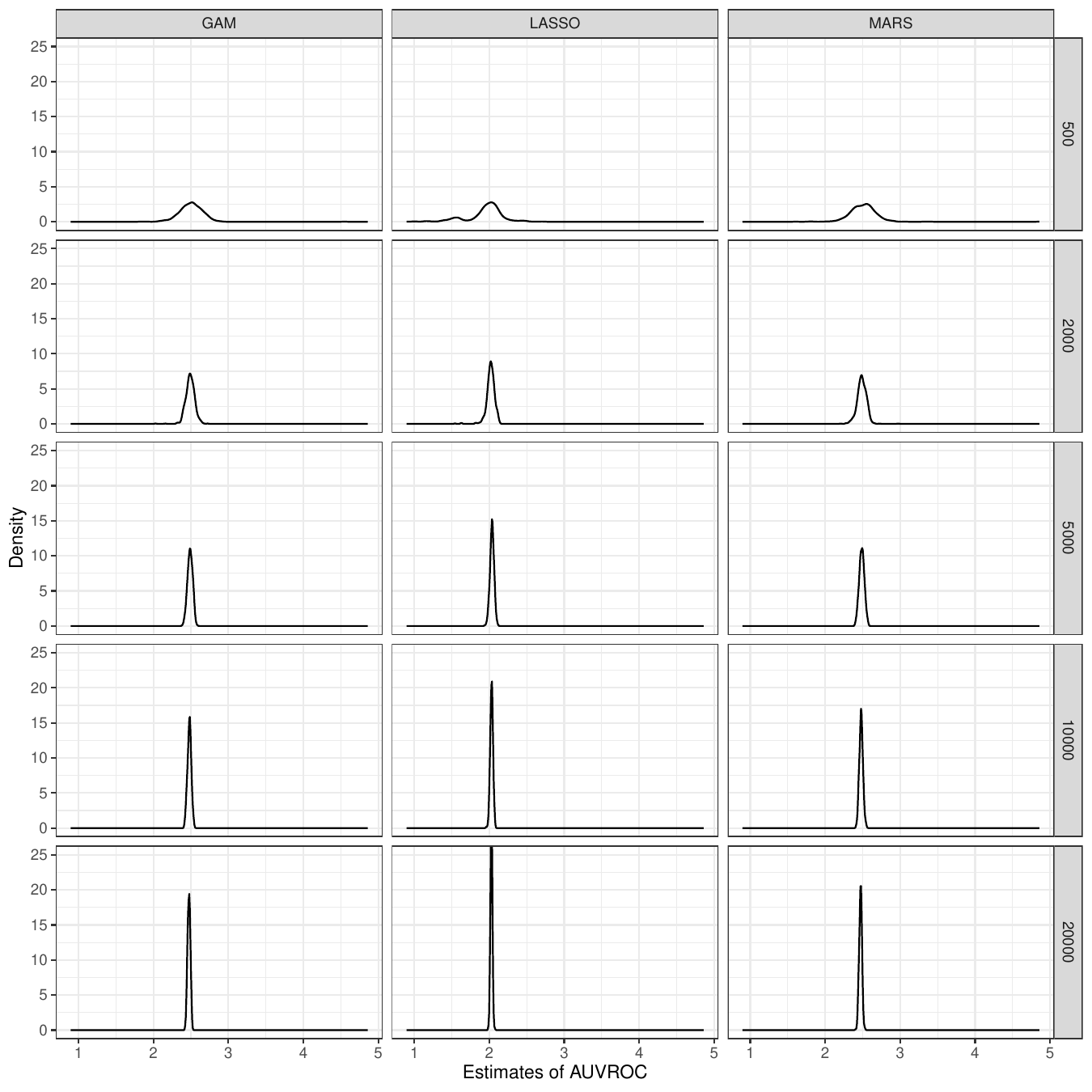}
    \caption{Estimated densities of the sampling distribution of the AUVROC estimators under different ranking algorithms and sample sizes from the numerical studies.}
    \label{fig:hist auc}
\end{figure}

\begin{figure}[!htbp]
    \centering
    \includegraphics[width=\linewidth]{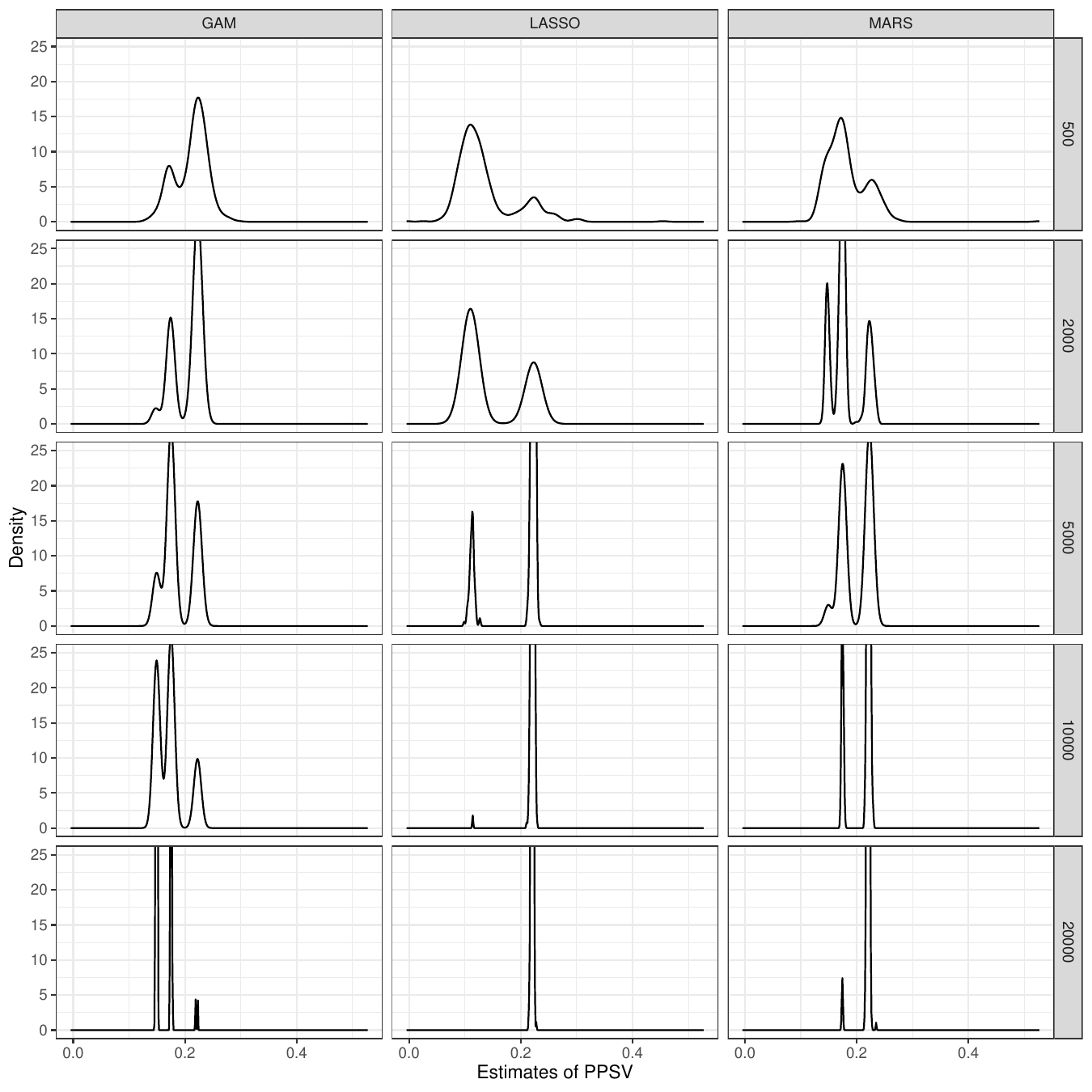}
    \caption{Estimated densities of the sampling distribution of the PPSV estimators under different selection algorithms and sample sizes from the numerical studies.}
    \label{fig:hist ave}
\end{figure}

\end{document}